\documentclass[twoside,journal]{IEEEtran}
\usepackage{amsfonts}
\usepackage{amssymb}
\usepackage{amsmath}
\usepackage{mathrsfs}
\usepackage{verbatim}
\usepackage{subfigure}
\usepackage{balance}
\usepackage{booktabs}
\usepackage{fancybox}
\usepackage{bm}
\usepackage{extarrows}
\usepackage{algorithm}
\usepackage{algorithmic}
\usepackage{multirow}
\usepackage{array}
\usepackage{graphicx}
\usepackage{epstopdf}
\usepackage{ulem}
\usepackage{cite}
\usepackage{subfig}

\newtheorem{proof}{Proof}
\newtheorem{proposition}{Proposition}

\newtheorem{definition}{Definition}

\begin{document}
\title{Secure Routing with Power Optimization for Ad-hoc Networks}
\author{Hui-Ming Wang, \IEEEmembership{Senior Member, IEEE,} Yan Zhang,  Derrick Wing Kwan Ng, \IEEEmembership{Senior Member, IEEE,}~and~Moon~Ho~Lee,~\IEEEmembership{Life~Senior~Member,~IEEE}

\thanks{H.-M. Wang and Y. Zhang are with the School of Electronic and Information Engineering, Xi'an Jiaotong University, and also with the Ministry of Education Key Laborato  ry for Intelligent Networks and Network Security, Xi'an 710049, China (e-mail: xjbswhm@gmail.com; yzhangxjtu@163.com).}
\thanks{D. W. K. Ng is with the School of Electrical Engineering and Telecommunications, University of New South Wales, Sydney, NSW 2052, Australia (e-mail: w.k.ng@unsw.edu.au).}
	\thanks{M. H. Lee is with the Division of Electronics Engineering, Chonbuk National 	University, Jeonju 561-756, Korea. (e-mail: moonho@jbnu.ac.kr).}
}
\maketitle
\begin{abstract}
In this paper, we consider the problem of joint secure routing and transmit power optimization for a multi-hop ad-hoc network under the existence of randomly distributed eavesdroppers following a Poisson point process (PPP).
Secrecy messages are delivered from a source to a destination through a multi-hop route connected by multiple legitimate relays in the network.
Our goal is to  minimize the end-to-end connection outage probability (COP) under the constraint of a secrecy outage probability (SOP) threshold, by optimizing the routing path and the transmit power of each hop jointly.
We show that the globally optimal solution could be obtained by a two-step procedure where the optimal transmit power has a closed-form and the optimal routing path can be found by Dijkstra's algorithm.
Then a friendly jammer with multiple antennas is applied to enhance the secrecy performance further, and the optimal transmit power of the jammer  and each hop of the selected route is investigated.
This problem can be 
solved optimally via an iterative outer polyblock approximation with one-dimension search algorithm.
Furthermore, suboptimal transmit powers can be derived using the successive convex approximation (SCA) method with a lower complexity.
Simulation results
show the performance improvement of the proposed algorithms for both non-jamming and jamming scenarios, and also reveal  a non-trivial trade-off between the numbers of hops and the transmit power of each hop for
secure routing.

\end{abstract}

\section{Introduction}

Ad-hoc networks have gained extensive research and analysis recent years due to the characteristics of self-organization and flexible networking \cite{ad_hoc_survey}.
However, because of the absence of a centralized administration and limited system resources, guaranteeing communication security in ad-hoc networks is quite challenging.
Specifically, in a multi-hop environment, since the information needs to be transmitted and relayed multiple times,
the threat from information leakage becomes higher
and the secrecy guarantee is quite difficult.
The traditional methods of interception avoidance are based on encryption technologies, which may not be applicable to emerging ad-hoc networks.
For instance, the time-varying network topologies require complicated key management which is hard to accomplish in decentralized networks.
Besides, the computing and processing abilities of the nodes may be limited and cannot afford the sophisticated encryption calculation.

On the other hand, physical layer security, an approach to achieve secrecy through the aspect of information theory by utilizing the characteristics of wireless channels,
has been widely studied for its advantages of low complexity and convenient distributed implementation \cite{Wyner-wiretap-channel, Survey-of-PLS}.
As a result,
various network models applying physical layer security have been investigated in the literature, such as interference channels \cite{if-channel1,if-channel2}, broadcast and multi-access channels \cite{broad-multi2,broad-multi3}, cooperative relay channels \cite{cooperate-relay2,cooperate-relay3,reviewer1_2},  and multi-antenna channels \cite{multi-antenna2,multi-antenna3,multi-antenna4,reviewer1_1}. The physical layer security approaches have also been introduced into the multihop ad-hoc or relaying networks. For instance,
authors in \cite{P1} 
compared three commonly-used relay selection schemes for a dual-hop network
with the constraints of security under the existence of eavesdroppers.
The authors in \cite{P2} 
proposed an optimal power allocation strategy for a predefined routing path to maximize the achievable secrecy rates under the constraint of maximum power budget. The authors in \cite{P3} 
 studied the secrecy and connection outage performance
under amplify-and-forward (AF) and decode-and-forward (DF) protocols for an end-to-end route, and discussed the trade-off between security and  QoS performance.
The authors in \cite{P4} 
explored the method to guarantee the network security
via routing and power optimization for a network with the deployment of cooperative jamming. In particular, \cite{P4}  assumed over-optimistically that
each jammer was located near one malicious eavesdropper to interfere the wiretapped information, which hardly be true in practice.


However, in the aforementioned works, perfect channel state information (CSI) and the locations of eavesdroppers are assumed to be available at the legitimate users, which is often impractical since the eavesdroppers usually work
passively and remain silent to hide their existence.
As a result, a general framework based on stochastic geometry was proposed to model the uncertainty of eavesdroppers' location \cite{summary_ppp}.
 In fact, Poisson point process (PPP) is the most widely adopted
distribution, and have been adopted in various researches studying network security, e.g.
\cite{ppp4,txzheng_artificial_noise,hmwang_pls_heterogeneous}.
Under the framework of stochastic geometry, in \cite{P6}
the authors 
investigated the secure routing problem.
The routing strategy aimed to achieve the highest secure connection probability under the DF relaying protocol.
The locations of eavesdroppers were assumed following the homogeneous PPP and both cases of colluding and non-colluding eavesdropping were considered.
Yet, optimal power allocation was not considered and it is unclear how the power allocation affects the system performance.

We have to point out here that for the routing security in a multi-hop network,
the performance of secure communication is coupled with
the numbers of hops and the transmit powers of each hop.
With higher (lower) transmit power, each hop can support a larger (smaller) transmission distance so that
less (more) numbers of routing hops are required to successfully relay messages to the destination. However, higher (lower) transmit power increases (decreases)  the probability of information leakage while less (more) numbers of routing hops decrease (increase) it.
Hence {\it{a non-trivial trade-off naturally exists for the secrecy performance between the numbers of hops and the transmit powers of each hop}}. However, none of the above works has revealed such an interesting trade-off, which is the main focus  of this paper. In particular,
in this paper, we investigate the secure routing and transmit power optimization problem in DF relaying networks
accompanied with PPP distributed eavesdroppers. The routing secrecy is evaluated by the minimum connection outage probability (COP)
subject to the constraint of secrecy outage performance.
The optimal route achieved the lowest COP is selected from all possible routing paths and the corresponding transmit power of each hop is also optimized. Moreover, friendly jamming is applied to further improve the security performance.
Different from previous studies, we a) consider secure routing under randomly distributed eavesdroppers and optimize the transmit powers jointly, and b) solve the power optimization problem with jamming using both the optimal monotonic optimization and the successive convex approximation (SCA) methods.
The main contributions are summarized as follows:

1) The secure routing design for a multi-hop network with the PPP distributed eavesdroppers is formulated as an optimization problem which minimizes the COP under a SOP constraint. The SOP and COP expressions for a given end-to-end path are derived in closed form. 
The closed-form expression of the optimal transmit powers is obtained.
By analyzing the expression of the minimum achieved COP for a given route and defining the routing weights, the routing problem can be interpreted as finding the route with the lowest sum weights, which can be solved optimally by the Dijkstra's algorithm.

2) A friendly jammer with multiple antennas is introduced to enhance the outage performance.
 For any fixed jamming power, the transmit powers allocation for the legitimate nodes on the obtained route is formulated as a monotonic optimization problem. The outer polyblock approximation with a one-dimension search algorithm is proposed to achieve the globally optimal solution. Later, to strike a balance between system performance and computational complexity, the SCA method is used to solve the problem considering its non-convexity feature. Though the solution derived from the SCA method is not globally optimal, the numerical results show that it achieves a close-to-optimal performance when it has a proper initial point.

3) The trade-off between the numbers of hops and the transmit powers of each hop is discussed for the routing security. The distribution of the numbers of the hops derived from simulation indicates that too many or few numbers of hops increase the leakage of information and rarely guarantee the security performance. This accentuates the importance of the joint consideration of transmit powers and secure routing.

The remainder of this paper is organized as follows. In Section II
, we present the system model of the multi-hop relaying network and formulate the secure routing as an optimization problem. In Section III, 
the routing and power optimization method
is provided.
The power optimization problem taking into account of friendly jamming is proposed in Section IV, then the outer polyblock approximation algorithm and the SCA algorithm are given in Section V.
Numerical results are presented in Section VI to illustrate the performance of the proposed algorithms.
The conclusions are summarized in Section VII.

The following notations are used in this paper. $(\cdot)^H$ and $|\cdot|$ represent Hermitian transpose and
absolute value, respectively. $\mathbb{P}(\cdot)$ denotes the probability and $E_A(\cdot)$ denotes the mathematical expectation with respect to A. $\mathcal{CN}(\mu,\sigma^2)$ represents the circularly symmetric complex Gaussian distribution with mean $\mu$ and variance $\sigma^2$.
The union and difference between two sets $\Omega_1$ and $\Omega_2$ are denoted by $\Omega_1\bigcup\Omega_2$ and $\Omega_1 \setminus\Omega_2$, respectively.

\section{System Model And Problem Description} \label{system_model}

We consider a multi-hop wireless ad-hoc network which consists of $M$ legitimate nodes \cite{P6}. The distribution of the eavesdroppers in the network follows the homogeneous PPP denoted as $\Phi$ with density $\lambda_e$. Each of the legitimate nodes and eavesdroppers is equipped with a single omnidirectional antenna.
One legitimate node aims to send messages to another in the network.
In order to transmit information from the source node to the destination node securely, a routing path needs to be found.
As we have mentioned before, there exists a non-trivial tradeoff between the numbers of hops and transmit powers.
Therefore the messages can be sent either directly to the destination with a high transmit power, or through multihop via several relays.
Assuming a routing path contains $N-1$ relay nodes,
each hop 
can be denoted by $l_n,n=1,...,N$, with the transmitter and the receiver at the $n$-th hop denoted as $T_n$ and $R_n$, respectively. Then, the entire routing path can be denoted by $\Pi=\{l_1,l_2,...,l_{N}\}$.
An illustration of the system model is shown in Fig. \ref{model}.



\begin{figure}[t]
\centering
\includegraphics[width=3.0in]{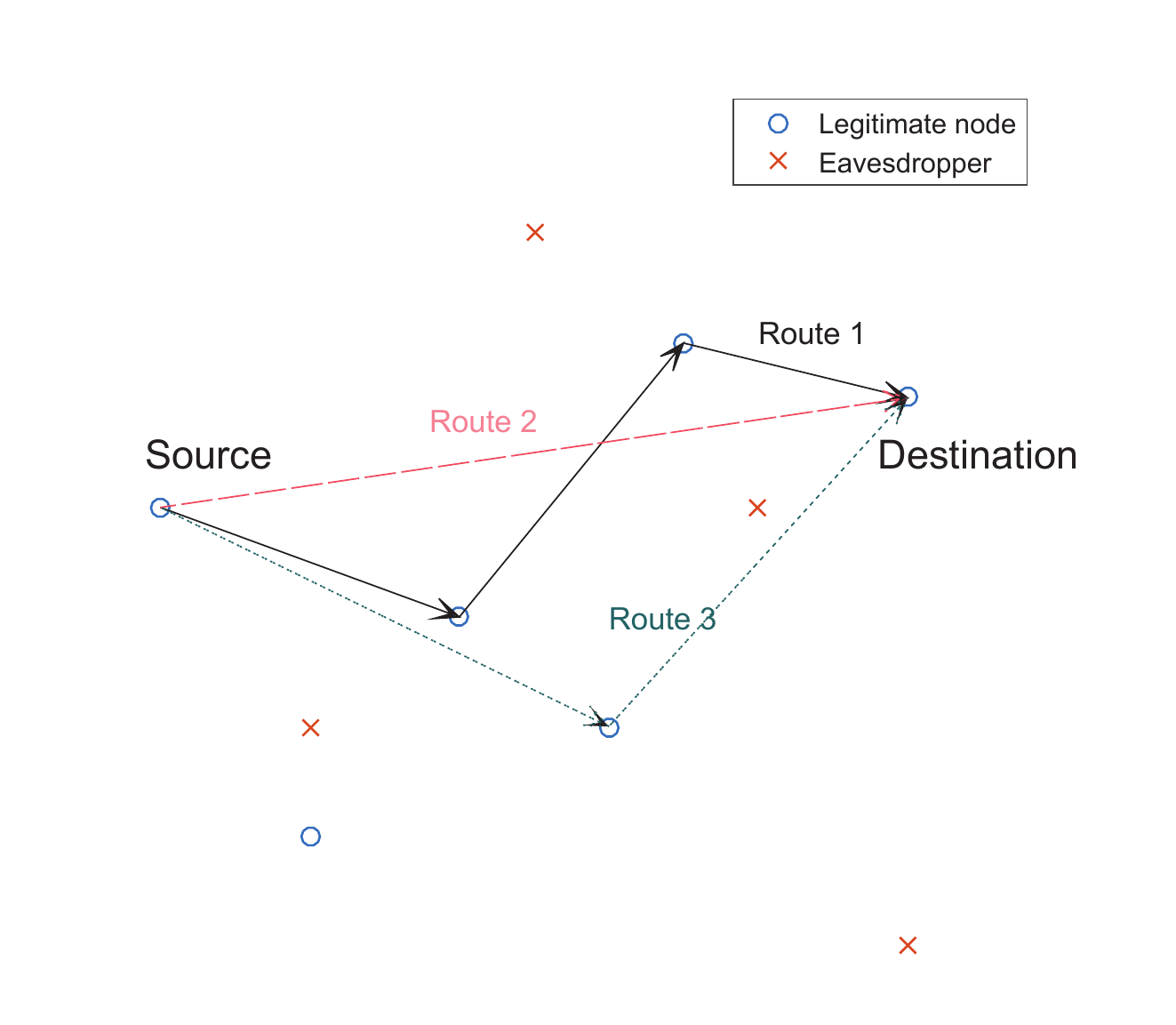}
\caption{The system model with multiple legitimate nodes and eavesdroppers.}
\label{model}
\end{figure}

The wireless channels are subjected to small-scale Rayleigh fading together with a large-scale path loss.
Each Rayleigh fading coefficient $h_{i,j}$ ($i$ and $j$ denote the transmitter and receiver of the path, respectively) is modeled as independent complex Gaussian with zero mean and unit variance, i.e., $h_{i,j}\sim\mathcal{CN}(0,1)$, and the path loss exponent is $\alpha$. We assume that the CSI and the locations of legitimate nodes are known while those of the eavesdroppers cannot be obtained because the eavesdroppers work passively.

Since each route $\Pi$ from the source to the destination is composed of several hops, under DF relaying scheme, a widely adopted protocol in literature \cite{P1,P2,P3,P4},
we first consider the transmission of hop $l_n$ from $T_n$ to $R_n$.
Let $x_{T_n}$ 
denote the symbol transmitted by $T_n$
, then the received signals at the legitimate receiver $R_{n}$ and the eavesdroppers are given by
%
\begin{align}
&y_{R_n} = \frac{\sqrt {P_{T_n}^{(\Pi)}} h_{{T_n}{R_n}}}{d_{{T_n}{R_n}}^{\alpha /2}}{x_{T_n}} + n_{R_n}\label{y_n} \  \  \text{and} \\
&y_{{E_i}} = \frac{\sqrt {P_{T_n}^{(\Pi)}} h_{{T_n}{E_i}}}{d_{{T_n}{E_i}}^{\alpha /2}}{x_{T_n}} + n_{E_i}, \label{y_e}
\end{align}
where $y_{R_n}$ and $y_{E_i}$ denote the received signals at receiver $R_n$ and eavesdropper $E_i\in\Phi$, respectively. $P_{T_n}^{(\Pi)}$ 
denotes the transmit power of node $T_n$ in route $\Pi$
. $d_{i,j}$ denotes the distance between $i$ and $j$. $n_{R_n}$ and $n_{E_i}$ are the noises at $R_n$ and $E_i$ following
$\mathcal{CN}(0,\sigma^2)$.

 Now we can derive the expressions of SNR at $R_n$ and $E_i$, which are given by
\begin{align}
&\mathrm{SNR}_{R_n}=\frac{P_{T_n}^{(\Pi)}|h_{{T_n}{R_n}}|^2}{d_{{T_n}{R_n}}^\alpha\sigma^2} \  \  \text{and}\\
&\mathrm{SNR}_{E_i}=\frac{P_{T_n}^{(\Pi)}|h_{{T_n}{E_i}}|^2}{d_{{T_n}{E_i}}^\alpha\sigma^2},
\end{align}
respectively.

In this paper, we adopt
 \textit{connection outage probability} (COP) and \textit{secrecy outage probability} (SOP) as performance metrics to measure the routing security. To improve security,
we let transmit nodes use different codebooks to retransmit the signal, so that the eavesdroppers cannot combine the wiretapped signals from multiple hops and could only decode these signals individually\footnote{The problem with colluding eavesdroppers requires the CDF of the sum of a number of independent but not identical distributed variables which subject to exponential distribution and whose number follows a PPP distribution, which is quite complicated. Due to the space limitation, here we only focus on the non-colluding cases.} \cite{TX_Z_DF_codebook}.
For an entire routing path, an end-to-end connection outage refers to the situation that the received SNR at any hop in the route is less than a predefined threshold $\gamma_c$, thus
the receiver cannot decode the message successfully and the corresponding probability of this event is called connection outage probability.
Secrecy outage occurs when the SNR of at least one eavesdropper at any hop surpasses the predefined threshold $\gamma_e$, hence the message can be intercepted by the eavesdropper(s). The probability of secrecy outage is called secrecy outage probability\footnote{The traditional SOP defined as $\mathbb{P}\{C_s<R_s\}$ is equivalent to our definition $\mathbb{P}\{SNR_e>\gamma_e\}$ when $\gamma_e\triangleq 2^{C_t -R_s}-1$.}.
The COP and SOP are denoted as $\mathcal{P}_\mathrm{co}$ and $\mathcal{P}_\mathrm{so}$, respectively.

 We consider the problem of finding the optimal routing path and transmit powers of each hop to achieve the lowest COP subject to
 a constraint that the SOP is no more than a predetermined value.
Denote $\Psi$ as the set of all feasible routing paths from the source to the destination and $\zeta$ as the maximum tolerable SOP, the optimization problem can be defined as:
\begin{equation} \label{Question}
\begin{split}
\min_{\Pi\in\Psi,P_{T_n}^{(\Pi)}}&\mathcal{P}_\mathrm{co}\\
\mathrm{s.t.}\ \mathcal{P}&_\mathrm{so}\leq\zeta.
\end{split}
\end{equation}
The objective function is equivalent to optimizing the route and transmit powers sequentially, that is
\begin{equation}\label{opt_prob}
\min_{\Pi\in\Psi,P_{T_n}^{(\Pi)}}\mathcal{P}_\mathrm{co}
=\min_{\Pi\in\Psi}\big(\min_{P_{T_n}^{(\Pi)}}\mathcal{P}_\mathrm{co}(\Pi)\big).
\end{equation}
Therefore, in the following section,
a secure routing method is proposed to solve this problem. The method can be divided into two parts: First, we optimize the transmit powers for any given route; then we find the optimal secure route from the source to the destination.

\section{Power optimization and Secure Routing} \label{power_and_route}

In this section, we study problem (\ref{Question})
and propose a method finding a secure route
with power optimization strategy
for the considered multi-hop network.
The closed-form expressions of COP and SOP are derived first and
then be used to facilitate the optimization of transmit power.
Finally the optimal secure routing path is obtained.
\subsection{Connection and Secrecy Outage Probabilities}

 First, we derive the exact expressions of COP and SOP for a given route.
 According to the definition of COP and the assumption of independent fading in Section II, the COP for route $\Pi$ denoted as $\mathcal{P}_\mathrm{co}(\Pi)$ can be written as
\begin{align}
\mathcal{P}_\mathrm{co}(\Pi)&=1-\prod_{l_n\in\Pi}\mathbb{P}\left\{\mathrm{SNR}_{R_n}>\gamma_c\right\}\nonumber\\
&=1-\prod_{l_n\in\Pi}\mathbb{P}\left\{\frac{P_{T_n}^{(\Pi)}|h_{{T_n}{R_n}}|^2}{d_{{T_n}{R_n}}^\alpha\sigma^2}>\gamma_c\right\}\nonumber\\
&=1-\prod_{l_n\in\Pi}\exp\left\{-\frac{\gamma_c\sigma^2d_{{T_n}{R_n}}^\alpha}{P_{T_n}^{(\Pi)}}\right\}, \label{pco_f}
\end{align}
where (\ref{pco_f}) holds since the fading coefficient $|h_{{T_n}{R_n}}|^2$ follows an exponential distribution with
$E\{|h_{{T_n}{R_n}}|^2\}=1$.

On the other hand, 
due to the usage of different codebooks at each hop and
since the distribution of eavesdroppers follows a homogeneous PPP, the SOP for route $\Pi$ denoted as
$\mathcal{P}_\mathrm{so}(\Pi)$ is given by
\begin{equation}
\mathcal{P}_\mathrm{so}(\Pi)=1-E_\Phi\left
\{\prod_{l_n\in\Pi}\prod_{E_i\in\Phi}\mathbb{P}(\mathrm{SNR}_{E_i}<\gamma_e)\right\}.
\end{equation}
To facilitate the derivation of a concise SOP expression,
the distribution of eavesdroppers for each hop is assumed to be uncorrelated with each other, which represents an upper bound of the original stationary eavesdroppers assumption \cite{independent_eves}.
Hence $\mathcal{P}_\mathrm{so}(\Pi)$ can be reformulated as
\begin{align}
&\mathcal{P}_\mathrm{so}(\Pi)=1-\prod_{l_n\in\Pi}\left\{E_\Phi\prod_{E_i\in\Phi}\mathbb{P}(\mathrm{SNR}_{E_i}<\gamma_e)
\right\}\nonumber \\
&=1-\prod_{l_n\in\Pi}\left\{E_\Phi\prod_{E_i\in\Phi}\left[1-\exp\left(-\frac{\gamma_e\sigma^2d_{{T_n}{E_i}}^\alpha}{P_{T_n}^{(\Pi)}}\right)\right]
\right\}\nonumber \\
&\overset{(a)}{=}1-\prod_{l_n\in\Pi}\left\{\exp\left[-\lambda_e\int_0^{2\pi}\int_0^\infty \exp\left(-\frac{\gamma_e\sigma^2r^\alpha}{P_{T_n}^{(\Pi)}}\right)rdrd\theta
\right]\right\}\nonumber\\
&\overset{(b)}{=}
1-\prod_{l_n\in\Pi}\exp\left\{
-\frac{\lambda_e 2\pi\Gamma(2/\alpha)}{\alpha}\left(\frac{\gamma_e\sigma^2}{P_{T_n}^{(\Pi)}}
\right)^{-\frac{2}{\alpha}}
\right\},\label{pso_f}
\end{align}
where $(a)$ holds for the probability generating functional lemma (PGFL) for the homogeneous PPP \cite{P12} under the assumption that the transmitter of each hop locates at the origin of the polar coordinate and $(b)$ holds for the
integration formula [\ref{integral}, 3.326.2] $\int_0^\infty x^m\exp(-\beta x^n)dx=\Gamma(\gamma)/(n\beta^\gamma)$ with $\gamma=(m+1)/n$.

Till now, we have obtained the expressions for  $\mathcal{P}_\mathrm{co}(\Pi)$ and $\mathcal{P}_\mathrm{so}(\Pi)$ in (\ref{pco_f}) and (\ref{pso_f}), respectively. 
The two formulas indicate that the powers of transmitters has opposite influence for COP and SOP. A higher power leads to less communication outage while a higher probability of information leakage.
Therefore, when study the COP and SOP performance jointly, this trade-off needs to be considered and a careful design is required for the transmit powers.
For the sake of conciseness, defining $\omega\triangleq\frac{2\pi\lambda_e}{\alpha}\Gamma\left(\frac{2}{\alpha}\right)(\gamma_e\sigma^2)^{-\frac{2}{\alpha}}$
and $\psi_n\triangleq\gamma_c d_{{T_n}{R_n}}^\alpha\sigma^2$,
(\ref{pco_f}) and (\ref{pso_f}) can be simplified as
\begin{align}
\mathcal{P}_\mathrm{co}(\Pi)&=1-\exp\left(-\sum_{l_n\in\Pi}\frac{\psi_n}{P_{T_n}^{(\Pi)}}\right)\label{PCO}\ \ \text{and}\\
\mathcal{P}_\mathrm{so}(\Pi)&=1-\exp\left[-\omega\sum_{l_n\in\Pi} \left(P_{T_n}^{(\Pi)}\right)^{2/\alpha
}
\right],\label{PSO}
\end{align}
respectively.

Based on (\ref{PCO}) and (\ref{PSO}), in the following part of this section, we try to solve (\ref{opt_prob}) under the constraint of SOP and propose the secure routing algorithm.

\subsection{Transmit Power Optimization}
Now we focus on optimizing the transmit powers of each hop to minimize the COP while satisfying the maximum tolerable SOP constraint.
The power optimization problem for any given route $\Pi$ can be written as
\begin{equation} \label{opt_prob1}
\begin{split}
\min_{P_{T_n}^{(\Pi)}}&\ \mathcal{P}_\mathrm{co}(\Pi) \\
\mathrm{s.t.}&\ \mathcal{P}_\mathrm{so}(\Pi)\leq\zeta.
\end{split}
\end{equation}
Substituting (\ref{PSO}) into the inequality constraint of (\ref{opt_prob1}), we have
\begin{equation}\label{inequal}
1-\exp\left[-\omega\sum_{l_n\in\Pi}\left(P_{T_n}^{(\Pi)}\right)^{2/\alpha
}
\right]\leq\zeta,
\end{equation}
which can be further transformed into
\begin{equation}
\omega\sum_{l_n\in\Pi} \left(P_{T_n}^{(\Pi)}\right)^{2/\alpha
}\leq\varepsilon\triangleq\ln\frac{1}{1-\zeta}.
\end{equation}

Notice that $\mathcal{P}_\mathrm{co}(\Pi)$ is a non-increasing function of $P_{T_n}^{(\Pi)}$. Since the expression on the left side of the inequality constraint is non-decreasing respect to $P_{T_n}^{(\Pi)}$,
this problem reaches its optimum when the inequality constraint is active at the optimal solution. As a result, we can safely replace
the inequality sign with an equality sign and
(\ref{opt_prob1}) can be rewritten as
\begin{equation}
\begin{split}
\min_{P_{T_n}^{(\Pi)}}&\ \sum_{l_n\in\Pi}\frac{\psi_n}{P_{T_n}^{(\Pi)}}\\
\mathrm{s.t.}&\ \sum_{l_n\in\Pi} \omega\left(P_{T_n}^{(\Pi)}\right)^{2/\alpha}=\varepsilon. \label{equa_cons}
\end{split}
\end{equation}

Problem (\ref{equa_cons}) is not convex since its constraint is not affine (except when $\alpha =2$, which represents propagation in free space). In order to reformulate (\ref{equa_cons}) into a convex form, defining $t_n\triangleq\left(P_{T_n}^{(\Pi)}\right)^{2/\alpha}$,  we have $P_{T_n}^{(\Pi)}=t_n^{{\alpha}/{2}},(P_n>0, t_n>0)$. Therefore,  (\ref{equa_cons}) can be rewritten as
\begin{equation}
\begin{split}
\min_{t_n}&\ \sum_{l_n\in\Pi}\frac{\psi_n}{t_n^{{\alpha}/{2}}}\\
\mathrm{s.t.}&\ \sum_{l_n\in\Pi} \omega t_n=\varepsilon. \label{equa_convex}
\end{split}
\end{equation}

Problem (\ref{equa_convex}) is a convex problem due to its convex objective function and affine equality constraint and its global optimum can be obtained.
Applying the Lagrange multiplier method associated with the equality constraint in (\ref{equa_convex}),
we have the following function:
\begin{align}
G\left(t_n,\xi\right)
=\sum_{l_n\in\Pi}\frac{\psi_n}{t_n^{{\alpha}/{2}}}+\xi\left[\omega\sum_{l_n\in\Pi} t_n-\varepsilon\right],
\end{align}
where $\xi$ is the Lagrange multiplier.
Then we set the partial derivatives of $G\left(t_n,\xi\right)$
respect to $t_n$ to zero, which yields
\begin{equation}\label{power1}
t_n=\left(\frac{\psi_n\alpha}{2\omega\xi}\right)^{\frac{2}{2+\alpha}}.
\end{equation}
Substituting (\ref{power1}) into the constraint in (\ref{equa_convex}), the expression of $\xi$ is derived as
\begin{equation} \label{xi}
\xi=\frac{\alpha\omega^{\frac{\alpha}{2}}}{2}\left[\frac{1}{\varepsilon}\sum_{l_n\in\Pi}\psi_n^{\frac{2}{2+\alpha}} \right]^{\frac{\alpha+2}{2}}.
\end{equation}
Then substitute (\ref{xi}) into (\ref{power1}), and we have
\begin{equation}
t_n=\psi_n^{\frac{2}{\alpha+2}}\left[
\frac{\omega}{\varepsilon}\sum_{l_k\in\Pi}\psi_k^{\frac{2}{2+\alpha}}
\right]^{-1}.
\end{equation}
Finally, using $P_{T_n}^{(\Pi)}=t_n^{{\alpha}/{2}}$, we derive the optimal transmit power of hop $l_n$ as
\begin{equation}\label{trans_power}
P_{T_n}^{(\Pi)}=\psi_n^{\frac{\alpha}{\alpha+2}}\left[
\frac{\omega}{\varepsilon}\sum_{l_k\in\Pi}\psi_k^{\frac{2}{2+\alpha}}
\right]^{-\frac{\alpha}{2}}.
\end{equation}

The influence of density of eavesdroppers $\lambda_e$ and SOP constraint $\zeta$ on the transmit power $P_{T_n}^{(\Pi)}$ can be observed from (\ref{trans_power}).
The increase of $\lambda_e$ and the decrease of $\zeta$ lead to the decrease of $P_{T_n}^{(\Pi)}$.
This result is comprehensible since the decrease of $P_{T_n}^{(\Pi)}$ will lower the risk of information leakage, hence to guarantee the security under the existence of more eavesdroppers and satisfy a more stringent SOP constraint.


So far, we have solved the inner optimization of (\ref{opt_prob}). We still need to find the secure route with the minimum COP
from all possible paths in the multi-hop network.

\subsection{Optimal Route Selection }
Since the transmit power of any route under the SOP constraint
has the form shown in (\ref{trans_power}) and $\mathcal{P}_\mathrm{co}(\Pi)$ is expressed as (\ref{PCO}), the secure routing problem can be rewritten as
\begin{equation}
\Pi^*=\arg\min_{\Pi\in\Psi}\left\{1-\exp\left[-\left (\frac{\omega}{\varepsilon}\right)^{\frac{\alpha}{2}}\left(
\sum_{l_n\in\Pi}(\psi_n)^{\frac{2}{2+\alpha}}
\right)^{\frac{\alpha}{2}+1}
\right]\right\},\nonumber
\end{equation}
which is equivalent to
\begin{equation}\label{weight_opt}
\Pi^*=\arg\min_{\Pi\in\Psi}\sum_{l_n\in\Pi}(\psi_n)^{\frac{2}{2+\alpha}}.
\end{equation}

Considering $\psi_n^{\frac{2}{2+\alpha}}$ to be the weight of hop $l_n$, expression (\ref{weight_opt}) can be interpreted as to find the optimal route $\Pi^*$ which has the minimum sum of weights. This can be solved by the Dijkstra's algorithm effectively.
Having obtained the optimal route and calculated the transmit powers for all transmission nodes using (\ref{trans_power}), the minimum COP can be obtained via (\ref{PCO}). The whole optimization procedure is shown in Algorithm \ref{algorithm_1}.


The relation of the optimal route and the system parameters of the network is worthy discussing. From (\ref{weight_opt}) we notice that the routing optimizing is determined by the weight $\psi_n$, which is independent of the information of eavesdroppers as well as the constraint on SOP.
This indicates that changing the density of eavesdroppers $\lambda_e$ in the network and SOP threshold $\zeta$ of the optimization problem does not impact the final selection of secure route. This can be interpreted from the following perspective. Since the distribution of eavesdroppers is homogeneous and the CSI along with the locations of eavesdroppers are unknown,
eavesdroppers appears homogeneously for any route path of the legitimate network.
As a result, the expected influence of eavesdroppers toward all options are equal, or we can say the information of eavesdroppers does not affect the optimal routing design.
Hence the constraint set by $\zeta$ which constrains the eavesdroppers has no affect on routing either.

\begin{algorithm}
\caption{Secure Routing Algorithm for Ad-hoc Network with PPP Distributed Eavesdroppers}
\label{algorithm_1}
\begin{algorithmic}[1]
\REQUIRE Network information, SNR thresholds $\gamma_c$ and $\gamma_e$, maximum tolerable SOP $\zeta$;
\ENSURE Optimal route $\Pi^*$, optimal transmit powers $P_{T_n}^*$, minimum connection outage probability $\mathcal{P}_\mathrm{co}^*(\Pi^*)$;\\
\STATE Calculate $(\psi_n)^{\frac{2}{2+\alpha}}$ for all possible transmission pairs and exploit them as weights;\\
\STATE Use the Dijkstra's algorithm to find the optimal route $\Pi^*$;\\
\STATE Calculate corresponding transmit powers $P_{T_n}^*$ for each transmitter on route $\Pi^*$ using (\ref{trans_power});\\
\STATE Calculate the minimal COP $\mathcal{P}_\mathrm{co}^*(\Pi^*)$ for route $\Pi^*$ using (\ref{PCO});\\
\RETURN $\Pi^*$, $P_{T_n}^*$, $\mathcal{P}_\mathrm{co}^*(\Pi^*)$;
\end{algorithmic}
\end{algorithm}

{\it{Computational complexity analysis:}} As we have derived the closed-form transmit power in (\ref{trans_power}),
 the vast majority of the computation complexity comes from the execution of the Dijkstra's algorithm, which has a complexity of $O(M^2)$ \cite{P11}. Thus the computation complexity of Algorithm 1 is in the order of $O(M^2)$.

\section{Transmission power optimization with friendly jamming}\label{jamming_problem}

Algorithm \ref{algorithm_1} provides us an effective way to obtain the secure route and the optimal transmit powers under the existence of random distributed eavesdroppers.
In order to enhance the security performance further, we now
consider the existence of  a friendly jammer \cite{reviewer3_1}.
This scenario is often feasible in practical applications. Take the device-to-device (D2D) communication system as an example. A cellular base station can work as a friendly jammer to interfere with the interception of eavesdroppers and assist the legitimate communication among the D2D users.
Hence in this section, based on the optimal secure route obtained from Algorithm 1, friendly jamming is introduced
 and the transmit power optimization problem of each transmit node on the secure route is reconsidered.

\begin{figure}
\centering
\includegraphics[width=3.0in]{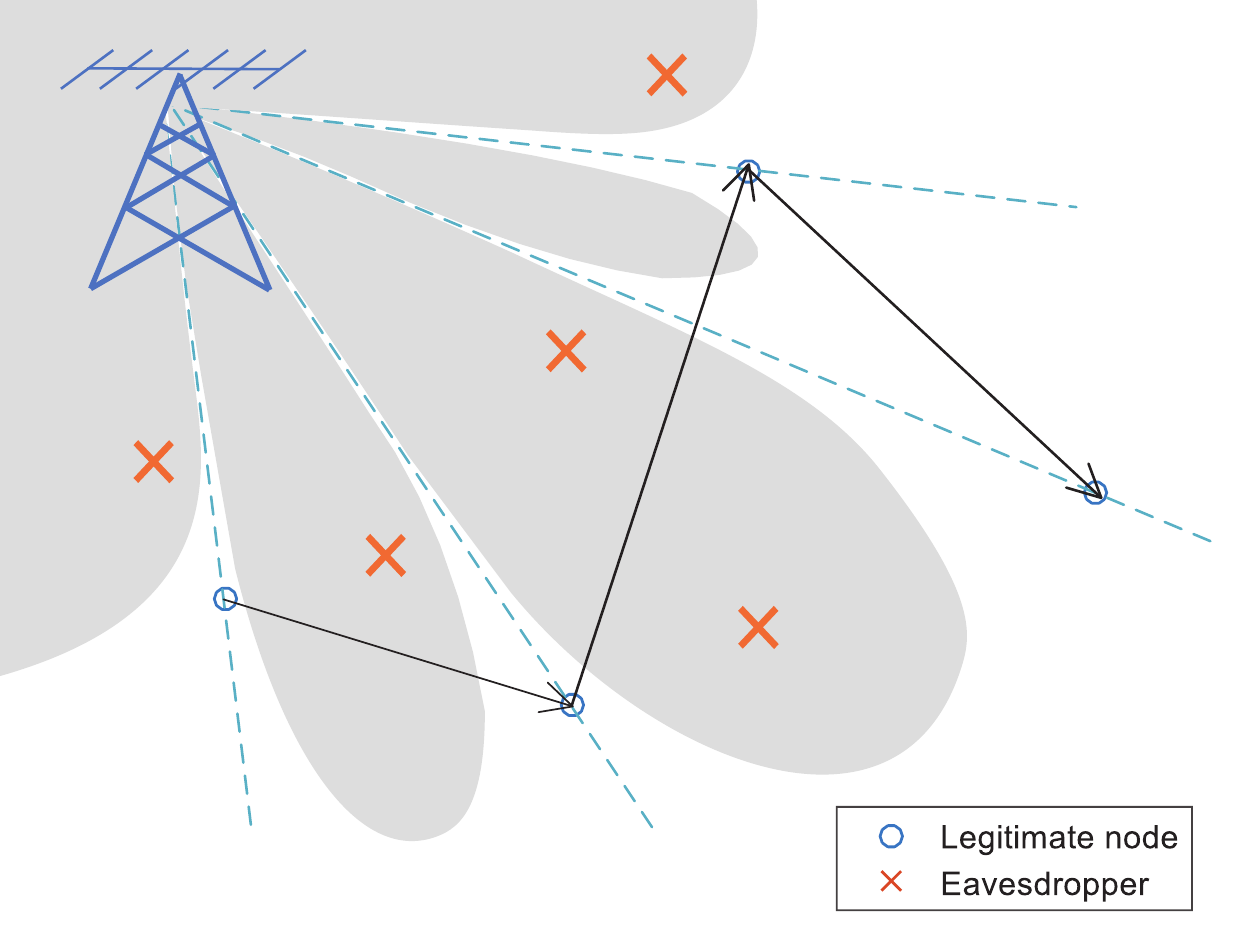}
\caption{Illustration of adopting friendly jamming for anti-eavesdropping.}
\label{jam_model}
\end{figure}

Based on the system model described in Section \ref{system_model}, we assume that there is a jammer equipped with multiple antennas in the network. The channels from the jammer and the transmitters to the legitimate receivers and to the malicious eavesdroppers are assumed uncorrelated with each other.
To secure the transmission, the friendly jammer radiates artificial noise
isotropically in the nullspace spanned by the channel vectors of the legitimate nodes to avoid interfering with the legitimate network.
The system model is depicted in
Fig. \ref{jam_model}.
With the knowledge of channel fading coefficients from the jammer to the legitimate receiver $R_n$, denoted as ${\mathbf{h}_{JR_{n}}},n=1,...,N$, the jammer adjusts its beamforming weight vector $\mathbf{v}$ to
suppress the artificial noise to legitimate receiver according to 
\begin{equation} \label{nullify}
\mathbf{h}^{H}_{JR_{n}}\mathbf{v}=0 ,\ n=1,...,N.
\end{equation}

 Denote the optimal route obtained from Algorithm \ref{algorithm_1} as $\Pi^*$, the transmit power for node $T_n$ in $\Pi^*$ as $P_{T_n}$, the transmit power of jammer as $P_J$, and the transmit power of jammer as $P_J$ with its beamforming weight vector normalized as $\Vert\mathbf{v}\Vert^2=1$. Due to (\ref{nullify}),
 the expression of COP is identical to (\ref{pco_f}).
 We assume that the interference produced by the jammer is much larger than the noise, then the noise at the eavesdroppers can be neglected
 and the expression of SOP is given by
\begin{equation}\label{pso_f_j}
\mathcal{P}_\mathrm{so}(\Pi^*)\!=\!1-\!\prod_{l_n\in\Pi^*}\!\left\{ E_{\Phi}\prod_{E_i}\left[\mathbb{P}\left(\frac{P_{T_n}|h_{{T_n}{E_i}}|^2/d^{2}_{{T_n}{E_i}}}{P_J|\mathbf{h}^{H}_{JE_{i}}\mathbf{v} |^2/d^{\alpha}_{JE_i}}\!<\!\gamma_e\right) \right] \right\}.
\end{equation}
In fact, expression (\ref{pso_f_j}) is an upper bound of the accurate SOP under jamming due to the neglect of noise, which represents a worst case of the exact value.

Due to 
$\Vert\mathbf{v}\Vert^2=1$ and the independence between $h_{{T_n}{E_i}}$ and $\mathbf{h}_{{J}{E_i}}$,
$|\mathbf{h}^H_{{T_n}{E_i}}\mathbf{v}|^2$ follows an exponential distribution with $E\{|\mathbf{h}^H_{{T_n}{E_i}}\mathbf{v}|^2\}=1$.
Therefore we have
\begin{align}
&\mathbb{P}\left(\frac{P_{T_n}|h_{{T_n}{E_i}}|^2/d^{2}_{{T_n}{E_i}}}{P_J|\mathbf{h}^{H}_{JE_{i}}\mathbf{v} |^2/d^{\alpha}_{JE_i}}<\gamma_e\right)\nonumber\\
=&1 - E_{|\mathbf{h}^{H}_{JE_i}\mathbf{v}|^2}\left[\exp\left(-\frac{\gamma_e P_J|\mathbf{h}^{H}_{JE_i}\mathbf{v}|^2/d^{\alpha}_{JE_i}}{P_{T_n}/d^{\alpha}_{T_nE_i}}\right) \right]\nonumber\\
=&1-\frac{1}{1+\frac{\gamma_e P_J/d^\alpha_{JE_i}}{P_{T_n}/d^\alpha_{T_nE_i}}}.\label{pro_pso}
\end{align}
Substituting (\ref{pro_pso}) into (\ref{pso_f_j}) and using PGFL, SOP can be expressed as
\begin{equation}\label{pso_simp}
\begin{split}
\mathcal{P}_{so}(\Pi^*)
=1-\!\prod_{l_n\in\Pi^*}\!\exp\left\{-\lambda_e\int_{R^2}{\frac{1}{1+\frac{\gamma_e P_J/d^\alpha_{JE_i}}{P_{T_n}/d^\alpha_{T_nE_i}}}dx_{E_i}}\right\}\\
\end{split}
\end{equation}
where $R^2$ denotes the distribution area of eavesdroppers,  and $x_{E_i}$ the location of $E_i$.
(\ref{pso_simp}) cannot be expressed in closed form due to the complexity of $d^\alpha_{T_nE_i}$ with respect to $x_{E_i}$.

The expressions of COP and SOP with friendly jamming
 have been derived in (\ref{pco_f}) and (\ref{pso_simp}). 
Then the transmit power optimization problem with the assistance of a multi-antenna jammer under a total transmit power constraint can be written as
\begin{equation}\label{opt_jam1}
\begin{split}
\min_{P_{T_n},P_J}\ &\left\{1-\prod_{l_n\in\Pi^*}\exp\left(-\frac{\gamma_c\sigma^2d_{{T_n}{R_n}}^\alpha}{P_{T_n}}\right)\right\}\\
\mathrm{s.t.}\ &1-\prod_{l_n\in\Pi^*}\exp\left[-\lambda_e \int_{R^2}\frac{1}{1+\frac{\gamma_e P_J/d^\alpha_{JE_i}}{P_{T_n}/d^\alpha_{T_nE_i}}}dx_{E_i}\right]\leq\zeta \\
&P_J+\sum_{l_n\in\Pi^*} P_{T_n}\leq P_{total}\ .
\end{split}
\end{equation}
Following the same procedures transforming (\ref{opt_prob1}) to (\ref{equa_cons}) and using $\psi_n$ 
for brevity, (\ref{opt_jam1}) is equivalent to
\begin{equation}
\begin{split}\label{opt_jam}
\max_{P_{T_n},P_J}&\ -\sum_{l_n\in\Pi^*}\frac{\psi_n}{P_{T_n}}\\
\mathrm{s.t.}&\ \sum_{{l_n}\in\Pi^*}\int_{R^2}\frac{1}{1+\frac{\gamma_e P_J/d^\alpha_{JE_i}}{P_{T_n}/d^\alpha_{T_nE_i}}}dx_{E_i}\leq \frac{\varepsilon}{\lambda_e}\\
&P_J+\sum_{l_n\in\Pi^*} P_{T_n}\leq P_{total}.
\end{split}
\end{equation}

Problem (\ref{opt_jam}) is not a convex optimization problem. 
Interestingly, however,
the objective function is monotonic respect to $P_{T_n}$,
under a fixed $P_J$, then 
(\ref{opt_jam}) turns to a monotonic optimization problem under a fixed $P_J$. 
Therefore, we propose an outer polyblock approximation algorithm to obtain the global optimal solution of the inner monotonic optimization problem under a fixed $P_J$, and the solution of (\ref{opt_jam}) can be derived by searching within the results obtained from the outer polyblock algorithm for different $P_J$.
Then we propose an SCA algorithm to reduce the complexity, at the price of obtaining a sub-optimal solution.


\section{Algorithms for Power Optimization with Jamming}

In this section, we present our methods
to solve problem (\ref{opt_jam}).
The method based on outer polyblock approximation and one-dimension search is proposed first,
and the SCA algorithm is put forward next to reduce the complexity.

\subsection{Outer Polyblock Approximation and One-dimention Search}

The power optimization for transmit nodes under a fixed $P_J$ is considered first. This problem can be written as:
\begin{equation}
\begin{split}\label{opt_jam_sub}
\max_{P_{T_n}}&\ -\sum_{l_n\in\Pi^*}\frac{\psi_n}{P_{T_n}}\\
\mathrm{s.t.}&\ \sum_{{l_n}\in\Pi^*}\int_{R^2}\frac{1}{1+\frac{\gamma_e P_J/d^\alpha_{JE_i}}{P_{T_n}/d^\alpha_{T_nE_i}}}dx_{E_i}\leq \frac{\varepsilon}{\lambda_e},\\
&\ \sum_{l_n\in\Pi^*} P_{T_n}\leq P_{total}-P_J,
\end{split}
\end{equation}
which is a monotonic optimization problem with respect to $P_{T_n}, n=1,...,N$.
We aim to solve
(\ref{opt_jam_sub}) via the outer polyblock approximation algorithm based on the theory of monotonic optimization theory.
The solution obtained through the proposed iterative algorithm reaches the global optimum \cite{global_optimum}.

Now we rewrite (\ref{opt_jam_sub}) into a canonical form of monotonic optimization.
In order to simplify the expressions, we define
\begin{equation}
g_n\left(P_{T_n}\right)\triangleq\int_{R^2}\frac{1}{1+\frac{\gamma_e P_J/d^\alpha_{JE_i}}{P_{T_n}/d^\alpha_{T_nE_i}}}dx_{E_i},
\end{equation}
 and rewrite the optimization variables as a transmit power vector $\mathbf{p}=\{P_{T_1},P_{T_2},...,P_{T_N}\}$,
while the power region $\mathcal{R}$ is defined by
\begin{align}\label{define_region}
\hspace{-3mm}\mathcal{R}\triangleq\bigcup\Big\{\mathbf{p}:\sum_{l_n\in\Pi^*}g_n(P_{T_n})\leq&\frac{\varepsilon}{\lambda_e},\sum_n P_{T_n}\leq P_{total}-P_J,\nonumber\\
&P_{T_n}\geq 0, n=1,...,N \Big\}.
\end{align}
Based on the above definitions,
problem (\ref{opt_jam_sub})
can be written in the following form:
\begin{equation} \label{opt_jam_vector}
\begin{split}
\max\ &U(\mathbf{p})\triangleq \sum_{l_n\in\Pi^*} -\frac{\psi_n}{P_{T_n}}\\
\mathrm{s.t.}\ &\mathbf{p} \in \mathcal{R}.
\end{split}
\end{equation}
In the sequel, we aim to solve (\ref{opt_jam_vector}). The polyblock algorithm is proposed to obtain its globally optimal solution.

{\it{1) Preliminaries:}}\ \
In this subsection, we explain that problem (\ref{opt_jam_vector}) is a monotonic optimization problem. 
First, several definitions are listed as follows to facilitate the presentation \cite{global_optimum,P10,polyblock_2,rui_zhang_polyblock}.


\begin{definition}
Given any two vectors $\mathbf{x},\mathbf{x'}\in \mathbb{R}^n$, $\mathbf{x'}\geq \mathbf{x}$ denotes that $x'_i\geq x_i,\forall i=1,...,n$. If $\mathbf{x'}\geq \mathbf{x}$ and $x'_i> x_i$, 
$\exists i=1,...,n$, we say $\mathbf{x'}$ \emph{dominates} $\mathbf{x}$;
If $x'_i > x_i,\forall i=1,...,n$, we say $\mathbf{x'}$ \emph{strictly dominates} $\mathbf{x}$ and write $\mathbf{x'}>\mathbf{x}$.
\end{definition}

\begin{definition}
Function $f: \mathbb{R}^n \rightarrow \mathbb{R}$ is
called an
\emph{increasing} function on $\mathbb{R}_+^n$ if for two vectors $\mathbf{x'},\mathbf{x}\in R^n_+$,
$f(\mathbf{x'}) \geq f(\mathbf{x})$ can be implied from $\mathbf{x'} \geq \mathbf{x}$.
Function $f$ is called \emph{strictly increasing} if
for any two vectors $\mathbf{x'} \neq \mathbf{x}$, $f(\mathbf{x'}) >f(\mathbf{x})$ can be implied from $\mathbf{x'} \geq \mathbf{x}$.
\end{definition}

\begin{definition}\label{normal}
Set $\mathcal{D}\in\mathbb{R}^n_+$ is a \emph{normal} set if for all $\mathbf{x}\in\mathcal{D}$,
any points $\mathbf{x'}$ dominated by $\mathbf{x}$ also belongs to $\mathcal{D}$.

\end{definition}

\begin{definition}
A point $\mathbf{x} \in \mathbb{R}^n_+$ is said to be an \emph{upper boundary point} of a compact normal set $\mathcal{D}$ if $\mathbf{x}\in \mathcal{D}$ and no point in $\mathcal{D}$ strictly dominates $\mathbf{x}$. All the upper boundary points of $\mathcal{D}$ constitute the \emph{upper boundary} of $\mathcal{D}$,
which is denoted by $\partial^+\mathcal{D}$.
\end{definition}

\begin{definition}
For vector $\mathbf{v}\in\mathbb{R}^n_+$, the hyper rectangle $\mathbf{[0,v]=\{x|0\leq x \leq v}\}$ is called a box with $\mathbf{v}$ being its vertex. The union of a finite number of boxes is referred to as a polyblock.
\end{definition}

Now, we provide some important results of optimization problems based on polyblock via the following proposition.

\begin{proposition}\label{max_on_vertice}
 A strictly increasing function $f(\mathbf{x})$
 reaches its maximal value over a polyblock at one vertice of the polyblock.
\end{proposition}

\begin{proof}
Suppose that $f(\mathbf{x})$ attains the global maximum at $\mathbf{x}$  which is not a vertex of the polyblock, then there must exists one vertex $\mathbf{x'}$ dominating $\mathbf{x}$, i.e, $\mathbf{x'}\geq \mathbf{x}$ and $x_i'>x_i, \exists i=1,...,n$.
$f(\mathbf{x'})>f(\mathbf{x})$ holds since $f(\mathbf{x})$ is strictly increasing, which is contradicted against the assumption that $\mathbf{x}$ reaches the optimum.
\end{proof}

Based on the definitions and the proposition above, we have the following proposition.
\begin{proposition}
 Optimization problem (\ref{opt_jam_vector}) is a monotonic problem which has an increasing objective function with respect to $\mathbf{p}$ and the power region $\mathcal{R}$ is a compact normal set.
\end{proposition}

\begin{proof}
It is clear that $U(\mathbf{p})$ and $g_n(P_{T_n})$ are both  increasing functions of $P_{T_n}$. Therefore, $\mathcal{R}$ is a normal set obviously accoding to definition \ref{normal} and the constraints in (\ref{define_region}) define $\mathcal{R}$ as a compact set.
\end{proof}

{\it{2) Outer Polyblock Generation:}}\ \
Proposition \ref{max_on_vertice} reveals that the maximum of an increasing function can be found via searching among the vertices of the polyblock. Thus for a monotonic optimization problem, we can gradually approach its region by iteratively generating a series of polyblocks
and find its maximum via searching.

In the following paragraphs, a method to generate the polyblocks 
is provided.
First, we aim to find the vertex achieving the maximal value of $U(\mathbf{z})$ on the polyblock.
We use $\mathcal{P}^{(k)}$ to denote the polyblock generated at the $k$-th iteration, $Z^{(k)}$ the vertex set of the polyblock $\mathcal{P}^{(k)}$,
 then
 the vertex maximizing $U(\mathbf{z})$ denoted as $\tilde{\mathbf{z}}^{(k)}$
can be found by searching in set $ Z^{(k)}$.

Then we
project $\tilde{\mathbf{z}}^{(k)}$ onto the upper boundary of $\mathcal{R}$ along the line segment through the origin to $\tilde{\mathbf{z}}^{(k)}$
and get the intersection point $\mathbf{r}^{(k)}$. 
Denoting the $n$th element of vector $\tilde{\mathbf{z}}^{(k)}$ as $\tilde{z}^{(k)}_n$ and the scaling parameter as $\delta_k,0\leq\delta_k\leq1$, the projection operation can be represented as solving the following optimization problem
\begin{equation}\label{compute_delta}
\begin{split}
\max_{\delta_k}\ &\sum_n\left(-\frac{\psi_n}{\delta_k \tilde{z}^{(k)}_n}\right) \\
\mathrm{s.t.}\ &\sum_n{g_n}\left(\delta_k \tilde{z}^{(k)}_n\right)\leq\frac{\varepsilon}{\lambda_e}\\
&\delta_k \sum_n\tilde{z}^{(k)}_n\leq P_{total}-P_J.
\end{split}
\end{equation}
The intersection point $\mathbf{r}^{(k)}$ can be calculated by $\mathbf{r}^{(k)}=\delta_k\tilde{\mathbf{z}}^{(k)}$,
and the new vertices adjacent to $\tilde{\mathbf{z}}^{(k)}$ are generated according to
\begin{equation}\label{compute_new_vertex}
\mathbf{z}^{(k),n}=\tilde{\mathbf{z}}^{(k)}-\left(\tilde{z}^{(k)}_n-r^{(k)}_n\right)\mathbf{e}_n,\ n=1,...,N,
\end{equation}
where $\mathbf{z}^{(k),n}$ denotes the $n$th new vertex generated at the $k$th iteration, $r^{(k)}_n$ denotes the $n$th element of $\mathbf{r}^{(k)}$ and $\mathbf{e}_n$ is the $n$th column of the identity matrix of size $N$. Then the new vertex set is defined by
\begin{equation}\label{vertex_generate}
Z^{(k+1)}=Z^{(k)}\setminus\tilde{\mathbf{z}}^{(k)}\bigcup\left\{\mathbf{z}^{(k),1},...,\mathbf{z}^{(k),N}\right\}.
\end{equation}
The new polyblock $\mathcal{P}^{(k+1)}$ is the union of the boxes defined by vertices in $Z^{(k+1)}$. An illustration of the generation procedure is depicted in Fig. \ref{polyblock}.

\begin{figure}
\centering
\includegraphics[width=3.0in]{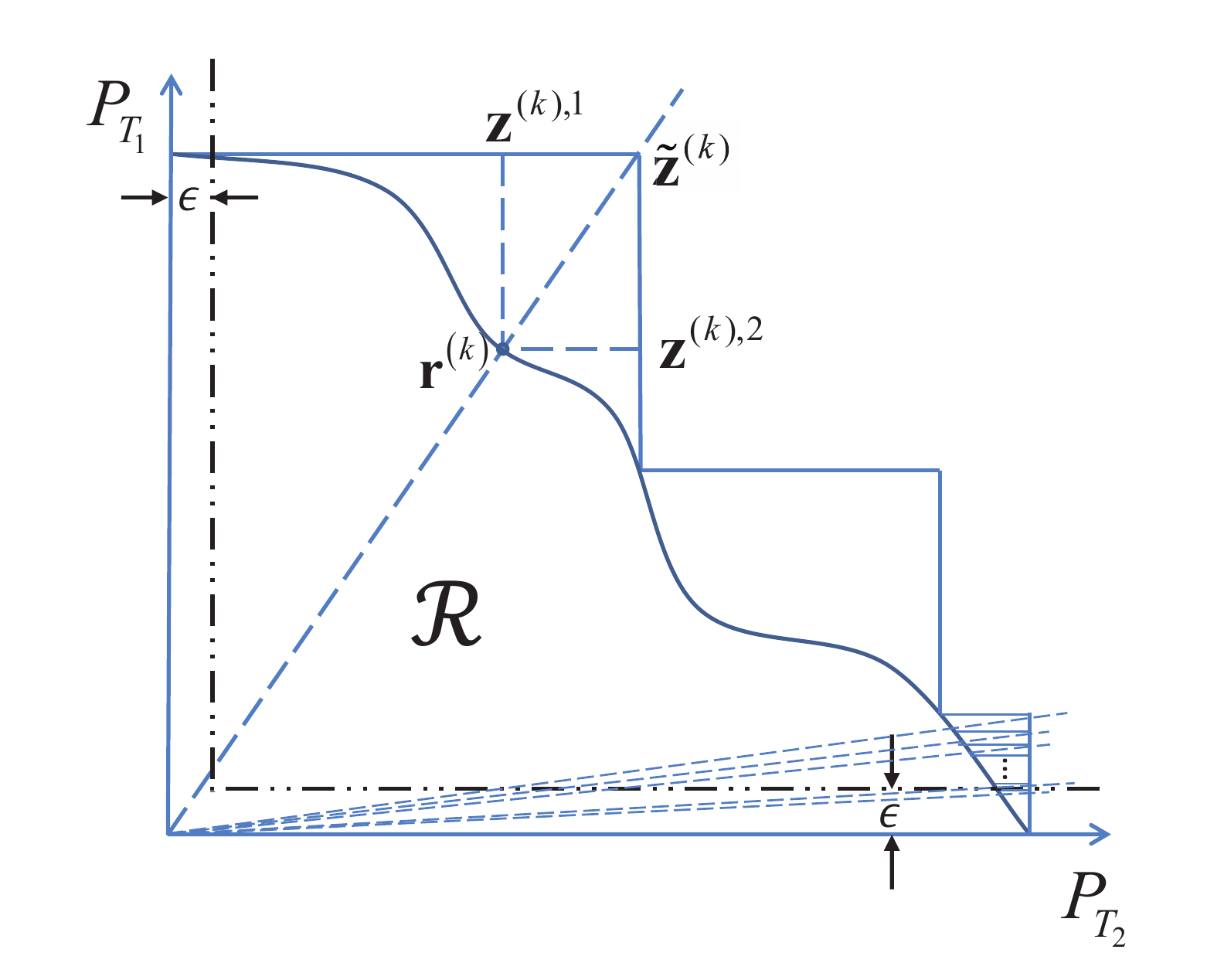}
\caption{Illustration of a procedure in polyblock generating when $\mathcal{R}$ is a noncovex but normal set.}
\label{polyblock}
\end{figure}

{\it{3) Outer Polyblock Approximation Algorithm:}}\ \
Based on the polyblock generation method, an iterative algorithm is proposed to obtain the optimal solution for problem (\ref{opt_jam_vector}).
The algorithm starts from calculating the initial vertex $\mathbf{z}^{(1)}$ for the first iteration. It is clear that the initial vertex should be the upper bound of the problem so that the box $[\mathbf{0}, \mathbf{z}^{(1)}]$ could cover the power region $\mathcal{R}$.
Obviously an upper bound is achieved when 
$\sum{g_n}\left( P_{T_n}\right)\leq \frac{\varepsilon}{\lambda_e}$ and $\sum P_{T_n}\leq P_{total}-P_J$ are relaxed for each item separately, which can be written more specifically as
\begin{equation} \label{relax_constr}
\begin{split}
\max_{P_{T_n}}&\ -\sum_{l_n\in\Pi^*}\frac{\psi_n}{P_{T_n}}\\
\mathrm{s.t.}&\  g_n(P_{T_n})\leq\frac{\varepsilon}{\lambda_e},\\
&\ P_{T_n}\leq P_{total}-P_J,\ n=1,...,N.
\end{split}
\end{equation}
The solution of (\ref{relax_constr}) acts as the initial vertex $\mathbf{z}^{(1)}$. Note that the selection of the initial point does not impact the final results since the solution of the outer polyblock approximation algorithm always converges to the global optimum.

In the $k$th iteration,
the optimal vertex $\tilde{\mathbf{z}}^{(k)}$ is first derived by
searching in vertex set $Z^{(k)}$, 
and the corresponding maximal value over polyblock $\mathcal{P}^{(k)}$ is denoted as
$U(\tilde{\mathbf{z}}^{(k)})$. 
Then the scaling parameter $\delta_k$ and the
intersection point $\mathbf{r}^{(k)}$ on the upper boundary of $\mathcal{R}$ is derived by solving (\ref{compute_delta}).
The optimal intersection point till the $k$th iteration $\tilde{\mathbf{r}}^{(k)}$
is obtained via 
\begin{equation}\label{best_intersection point}
\tilde{\mathbf{r}}^{(k)}= \arg\max\{U(\mathbf{r}^{(k)}),U(\tilde{\mathbf{r}}^{(k-1)})\}.
\end{equation}
$U(\tilde{\mathbf{z}}^{(k)})$ and  $U(\tilde{\mathbf{r}}^{(k)})$ are the upper and lower bound of the optimal value $U(\mathbf{r^*})$ respectively, thus
$U(\mathbf{r}^*)-U(\tilde{\mathbf{r}}^{(k)})<U(\tilde{\mathbf{z}}^{(k)})-U(\tilde{\mathbf{r}}^{(k)})$
is always satisfied.
 If $U(\tilde{\mathbf{z}}^{(k)})-U(\tilde{\mathbf{r}}^{(k)})$ is lower than a predefined number $\eta$,
$U(\mathbf{r}^*)$ is greater than $U(\tilde{\mathbf{r}}^{(k)})$ by no more than $\eta$. We quit the iteration and called
 $\tilde{\mathbf{r}}^{(k)}$ an $\eta$-optimal solution to problem (\ref{opt_jam_vector}). Otherwise, a new polyblock $\mathcal{P}^{(k+1)}$ is generated and the above procedure is repeated till
an $\eta$-optimal solution is obtained.

\begin{algorithm}
  \caption{The polyblock algorithm for transmit power optimization problem}
  \label{poly_algorithm}
  \begin{algorithmic}[1]
  \REQUIRE $\psi_n$, accuracies $\epsilon$ and $\eta$
  \ENSURE $\tilde{P}_{T_n}^*$
  \STATE Initialization: Set $k=1$.
  Find the initial vertex $\mathbf{z}^{(1)}$ by solving
  (\ref{relax_constr}) and set the initial vertex set as $Z^{(1)}=\{\mathbf{z}^{(1)}\}$;

  \WHILE{$(\epsilon,\eta)$-accuracy is not satisfied,}
  \STATE Obtain the optimal vertex $\tilde{\mathbf{z}}^{(k)}$ by searching in set $Z_{\epsilon}^{(k)}$, and compute $U(\tilde{\mathbf{z}}^{(k)})$; 
  \STATE Obtain the interection point $\mathbf{r}^{(k)}$ on the upper boundary of $\mathcal{R}$ by (\ref{compute_delta}); 
  \STATE Obtain 
  $\tilde{\mathbf{r}}^{(k)}$ by (\ref{best_intersection point}) and compute $U(\tilde{\mathbf{r}}^{(k)})$;
  \IF{$U(\tilde{\mathbf{z}}^{(k)})-U(\tilde{\mathbf{r}}^{(k)})\leq\eta$,}
  \STATE $\tilde{\mathbf{r}}^{(k)}$ is an $(\epsilon,\eta)$-optimal solution, $\tilde{P}_{T_n}^*=\tilde{r}^{(k)}_{n}, n=1,...,N$;
  \ELSE
  \STATE Generate the $(k+1)$-th polyblock $\mathcal{P}^{(k+1)}$ and vertex set $Z_{\epsilon}^{(k+1)}$ by (\ref{compute_new_vertex}) and (\ref{vertex_generate});
  \ENDIF
  \STATE $k=k+1$;
  \ENDWHILE
  \end{algorithmic}
\end{algorithm}

Suppose that the optimal solution is located in the region defined by $\{\mathbf{r}^*|0\leq r^*_{n}\leq \epsilon, 1\leq n \leq N\}$ with $\epsilon$ being a small positive number, then the polyblock algorithm would converges with a fairly low speed as $\tilde{\mathbf{z}}^{(k)}$ gradually approaching this region, as depicted in Fig.\ref{polyblock}. Thus, in order to guarantee the convergence speed of the algorithm, we replace the region $Z^{(k)}$ by $Z_{\epsilon}^{(k)}\triangleq \{\mathbf{z}\in Z^{(k)}|z_k\geq\epsilon,\forall k\}$. The parameter $\epsilon $
reflects the tradeoff between the accuracy and computational complexity.

The procedure for solving (\ref{opt_jam_vector}) is summarized in Algorithm \ref{poly_algorithm}.
The convergence explanation can be found in  [\ref{math_poly_algo}, Theorem 1].
 Given the accuracies $\epsilon$ and $\eta$, the proposed algorithm will terminate after a finite number of iterations and an $(\epsilon,\eta)$-approximate optimal solution for problem (\ref{opt_jam_vector}) can be derived.

So far, we have solved (\ref{opt_jam_sub}) through
 outer polyblock approximation algorithm
 and obtained the optimal powers of route $\Pi^*$
under a fixed $P_J$. 
By varying the value of $P_J$, a series of solutions for (\ref{opt_jam_sub}) under different $P_J$ can be derived, and the solution of (\ref{opt_jam}) can be derived through searching within these solutions.

{\it{ Computational complexity analysis:}}
It is clear that the searching precision of $P_J$ is influential to the computational complexity.
The complexity of the polyblock algorithm is sensitive to the values of $\epsilon$ and $\eta$ which influence the number of iterations of Algorithm 2.
In each iteration, the accuracy requirement of the bisection method for finding the projection point of the best vertex on the upper boundary $\mathcal{R}$ also affects the complexity.
The number of hops $N$, which represents the dimensionality of the problem, is another key factor.
Since the
size of vertex set increases $(N-1)$ after each iteration, the time of obtaining the optimal vertex also increases linearly, leading to a larger complexity for one iteration.
At the $k$th iteration, the optimal vertex $\tilde{\mathbf{z}}^{(k)}$ needs to be found from $(kN-k-N+2)$
vertices.
Supposing the polyblock algorithm reaches the $(\epsilon,\eta)$-optimal solution at the $\tilde{k}$th iteration, the complexity for obtaining the optimal vertices is in the order of $O(N\tilde{k}^2)$. Then the total complexity considering the one-dimension search among $\kappa$ different values of $P_J$ is in the order of $\kappa\times O(N\tilde{k}^2)$.


\subsection{Successive Convex Approximation Algorithm}

Though we have solved (\ref{opt_jam})
through an outer polyblock approximation with one-dimension search method, obtaining the global optimum is time-consuming. In this section, we use a sequence of convex problems to approximate this non-convex optimization problem (\ref{opt_jam}) based on the SCA method and solve the problem efficiently. The solution obtained from the SCA method turns out to be a local optimum \cite{reviewer3_2}.

First, we define
$F_n(x_{E_i})\triangleq\frac{\gamma_e d^\alpha_{T_nE_i}}{d^\alpha_{JE_i}}$ for brevity and set $c_n\triangleq\frac{P_{T_n}}{P_J}$, and (\ref{opt_jam}) can be rewritten as
 \begin{subequations}\label{sca_change_variable}
\begin{align}
\min_{c_n,P_J}&\ \sum_{l_n\in\Pi^*}\frac{\psi_n}{c_n P_J}\\
\mathrm{s.t.}&\ \sum_{{l_n}\in\Pi^*}\int_{R^2}\frac{1}{1+\frac{1}{c_n}F_n(x_{E_i})}dx_{E_i}\leq \frac{\varepsilon}{\lambda_e},\label{sca_sop_constraint}\\
&\ 1+\sum_{l_n\in\Pi^*} c_n\leq \frac{P_{total}}{P_J}.\label{sca_power_constraint}
\end{align}
\end{subequations}
Next, by introducing slack variables $a$ and $b$ and using $ab=\frac{1}{4}\left[(a+b)^2-(a-b)^2\right]$, we consider the following problem
\begin{subequations}\label{sca_slack_variable}
\begin{align}
\min_{a, b, c_n, P_J}&\  \frac{1}{4}\left[(a+b)^2-(a-b)^2\right]\\
\mathrm{s.t.}&\ a\geq \frac{1}{P_J},\label{sca_constraint_a}\\
&\  b\geq \sum_{l_n\in\Pi^*}\frac{\psi_n}{c_n},\label{sca_constraint_b}\\
&\ \mathrm{(\ref{sca_sop_constraint})}, \mathrm{(\ref{sca_power_constraint})}.\notag
\end{align}
\end{subequations}
Problems (\ref{sca_change_variable}) and (\ref{sca_slack_variable}) are equivalent at the optimal solution, since both the inequality constraints (\ref{sca_constraint_a}) and (\ref{sca_constraint_b}) are active at the optimal solution, otherwise $a$ and $b$ could be decreased to obtain a lower objective value. Therefore, we focus on solving  (\ref{sca_slack_variable}) in the following.

The terms $-(a-b)^2$, $\sum_{n}\int_{R^2}\frac{1}{1+F_n(x_{E_i})/c_n}dx_{E_i}$, and $-P_{total}/P_J$ are non-convex and problem (\ref{sca_slack_variable}) is difficult to be solved.
Denoting the optimal solutions of the convex approximation problem at the $(k-1)$th iteration as $a^{(k-1)}, b^{(k-1)}, P_J^{(k-1)} $, and $ \mathbf{c}^{(k-1)}=\left[c_1^{(k-1)},...,c_N^{(k-1)}\right]^T$,
we use the first-order Taylor expansion
around $a^{(k-1)}, b^{(k-1)}, P_J^{(k-1)} $, and $ \mathbf{c}^{(k-1)}$ to approximate the non-convex terms
 and construct a convex optimization problem \cite{sca_cwang}, which can be written as
\begin{align}
\hspace{-3mm}H_1(a,a^{(k-1)},b,b^{(k-1)}) & \triangleq \left(a^{(k-1)}-b^{(k-1)}\right)^2\nonumber\\
&-2\left(a-b\right)\left(a^{(k-1)}-b^{(k-1)}\right),
\end{align}
\begin{align}
\hspace{-2mm}H_2(\mathbf{c},\mathbf{c}^{(k-1)})\triangleq&
\sum_n\int_{R^2}\Bigg[\frac{c_n^{(k-1)}}{c_n^{(k-1)}+F_n(x_{E_i})}\\
&+\frac{\big(c_n-c_n^{(k-1)}\big)F_n(x_{E_i})}{\big(c_n^{(k-1)}+F_n(x_{E_i})\big)^2}\Bigg]dx_{E_i}, \ \mathrm{and}
\end{align}
\begin{flalign}
H_3(P_J,P_J^{(k-1)})\triangleq-\frac{2P_{total}}{P_J^{(k-1)}}+\frac{P_{total}P_J}{{P_J^{(k-1)}}^2},&&
\end{flalign}
respectively.
Therefore, the convex approximation problem is constructed as
\begin{equation}\label{sca_prob}
\begin{split}
\min_{a, b, \mathbf{c}, P_J}&\  \frac{1}{4}\left[(a+b)^2+H_1(a,a^{(k-1)},b,b^{(k-1)})\right]\\
\mathrm{s.t.}&\ a\geq \frac{1}{P_J},\  b\geq \sum_{l_n\in\Pi^*}\frac{\psi_n}{c_n},\\
&\ H_2(\mathbf{c},\mathbf{c}^{(k-1)})\leq \frac{\varepsilon}{\lambda_e},\\
&\ 1+\sum_{l_n\in\Pi^*} c_n+H_3(P_J,P_J^{(k-1)})\leq 0.
\end{split}
\end{equation}
Starting from one feasible initial point and solving (\ref{sca_prob}) iteratively till convergence is reached, 
a sub-optimal solution of problem (\ref{sca_slack_variable}) can be obtained.
The entire procedure is summarized in Algorithm 3.
%

{\it{Initial point selection:}}
The initial point of the SCA method should be a feasible solution of problem (\ref{sca_slack_variable}).
To find a feasible initial point efficiently, 
the outer polyblock approximation algorithm can be utilized. With a random jamming power $P_J$ and the corresponding solutions of Algorithm 2 \{$\tilde{P}_{T_n}^*, n=1,...,N$\}, an initial point of the SCA method can be generated according to
$P_J^{(1)}=P_J$,
$\{c_n^{(1)}=\tilde{P}_{T_n}^*/P_J^{(1)},n=1,...,N\}$,
$a^{(1)}\geq 1/P_J^{(1)}$, and
$b^{(1)}\geq \sum_{l_n\in\Pi^*}\psi_n/c_n^{(1)}$.


\begin{algorithm}
  \caption{Successive convex approximation algorithm for (\ref{sca_slack_variable})}
  \label{sca_algorithm}
  \begin{algorithmic}[1]
  \STATE {\bf{Initialization}}: Set an initial feasible point\\
  $(a^{(0)},b^{(0)},\mathbf{c}^{(0)},P_J^{(0)})$, accuracy $\rho$, and $k=0$;
  \WHILE{the difference of the optimal values of two successive iterations surpasses $\rho$,}
  \STATE $k=k+1$;
  \STATE Compute the solution of (\ref{sca_prob});
  \ENDWHILE
  \STATE $\bar{P}_{J}^*=P_J^{(k)},\bar{P}_{T_n}^*=c_n^{(k)}P_J^{(k)},n=1,...,N$;
  \end{algorithmic}
\end{algorithm}

{\it{Feasibility and convergence:}}
The solution of (\ref{sca_prob}) belongs to the feasible set of problem (\ref{sca_slack_variable}).
Since  $\sum_{n}\int_{R^2}\frac{1}{1+F_n(x_{E_i})/c_n}dx_{E_i}\leq H_2(\mathbf{c},\mathbf{c}^{(k-1)})$, and $-P_{total}/P_J\leq H_3(P_J,P_J^{(k-1)})$ are satisfied
based on the first-order condition \cite{convex_optimization}, the feasible set of problem (\ref{sca_prob}) is a convex subset of (\ref{sca_slack_variable}).
Besides, since the optimal solution at the $(k-1)$th iteration is a feasible solution for problem (\ref{sca_slack_variable}) at the $k$th iteration, the optimized objective value at the $k$th iteration should not surpass that at the $k$th iteration.
From the non-negative and non-increasing objective values obtained from the iterations, the convergence of the proposed SCA algorithm can be concluded.

{\it{Complexity and performance:}}
Problem (\ref{sca_prob}) can be reformulated as an SDP problem which has a computational complexity in the order of $O(N^{3.5})$ \cite{sdp_complexity}. Then the complexity of Algorithm 3 is $\bar{k}\times O(N^{3.5})$, where $\bar{k}$ denotes the number of iterations and is about few tens.
Comparing to the complexity of the polyblock approximation with one-dimension search method $\kappa\times O(N\tilde{k}^2)$,
in which $\tilde{k}$ is few hundreds and $\kappa$ is usually large for the searching accuracy, the complexity using SCA method is much smaller.
The selection of the initial point will impact the final results due to the locally optimal solution produced by the SCA method.
To improve the performance and decrease the influence, Algorithm 3 can be executed for a number of different initial points, and the solution obtaining the minimum will be select as the final optimized outcome.
Simulation shows that
operating the SCA algorithm for only few tens of initial points, the selected minimum is quite close to the optimum generated from the outer polyblock approximation with one-dimension search among thousands of different values of $P_J$.



\section{Numerical Results}

 In this section, numerical results are presented to illustrate the validity  and performance of the proposed schemes. We assume that $M=10$ legitimate nodes are deployed in a $20\times 20\ \mathrm{m}^2$ square area while the eavesdroppers are randomly distributed in the area of size $400\times 400\ \mathrm{m}^2$.
 We set SNR thresholds $\gamma_c=0.8$ dB and $\gamma_e=0$ dB. The path loss exponent $\alpha =4$. 
 The density of PPP distributed eavesdroppers is $\lambda_e=10^{-4}$. 

\subsection {Comparison Between Polyblock with One-dimention Search and SCA Algorithms}


\begin{figure}
\centering
\includegraphics[width=3.0in]{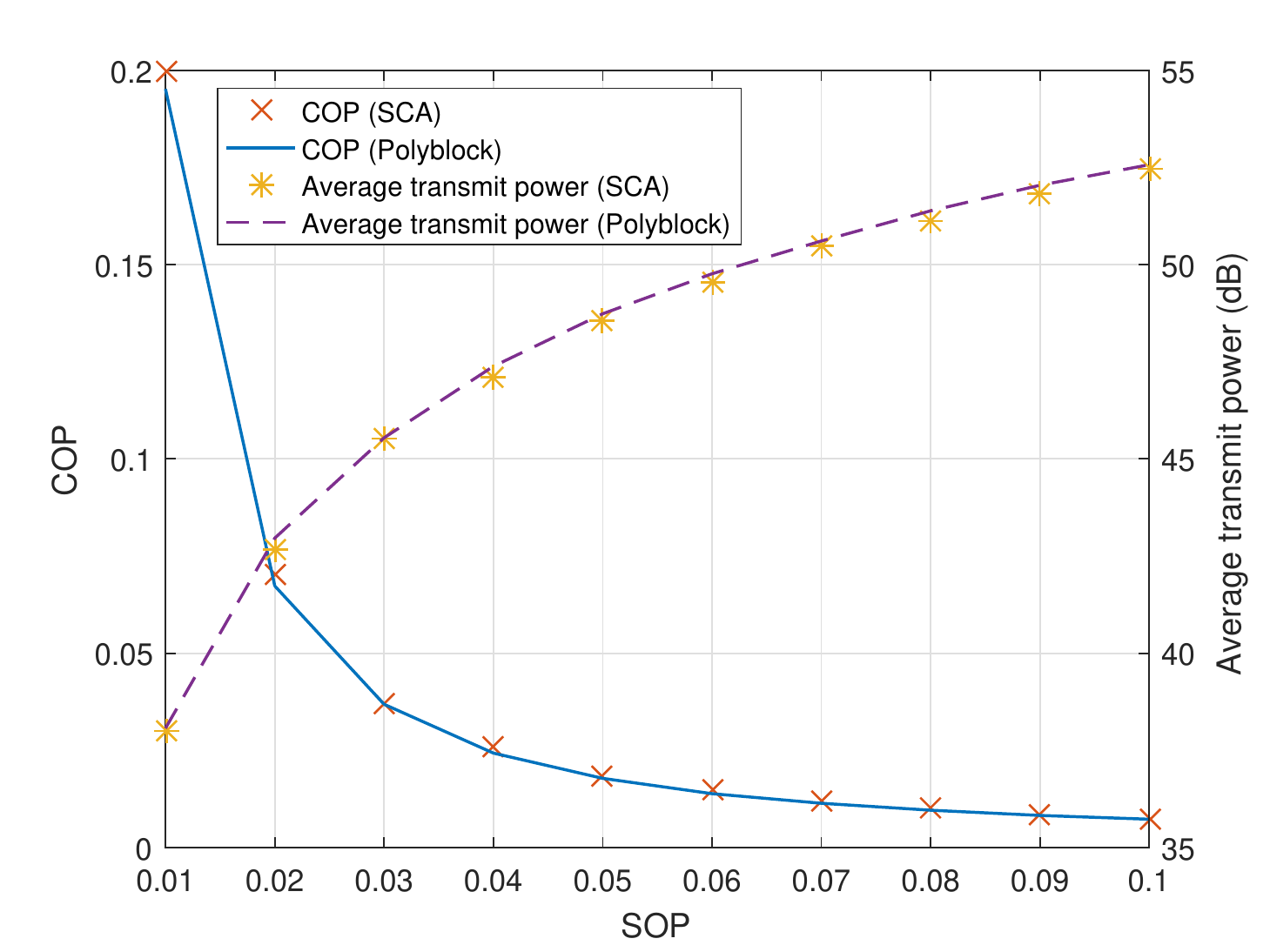}
\caption{COP and average transmit power versus SOP with friendly jamming.}
\label{sca_poly}
\end{figure}


Under the assistance of friendly jamming and the constraint of the total transmit power, the COP
and the average power of legitimate transmit nodes, which is defined as $\frac{1}{N}\sum{P_{T_{n}}}$ (normalized by the noise power $\sigma^2$),
obtained from the polyblock approximation with one-dimension search and the SCA algorithm with multiple initial points are compared in Fig. \ref{sca_poly}.
The polyblock algorithm is executed for $1000$ values of $P_J$, while the SCA algorithm is executed for $20$ different initial points. 
With the increase of SOP, the COP decreases while the average transmit power of the legitimate nodes increases, which is in accordance with the expressions of COP and SOP using friendly jamming derived in (\ref{pco_f}) and (\ref{pso_simp}). Specifically, (\ref{pco_f}) and (\ref{pso_simp}) indicate that COP is a decreasing function while SOP is an increasing function of the transmit powers. Therefore as the SOP becomes higher, the average transmit power increases accordingly and hence the COP decreases.
In fact, the relaxation of secure outage performance permits the transmitters to achieve a better connection performance at the expense of a higher transmit powers.
On the other hand, it is clear from the simulation that though theoretically the SCA algorithm obtains a suboptimal solution,
by executing the SCA algorithm for only a few different initial points and selecting the minimum as the final result,
the solutions are close to that generated from the outer polyblock approximation with one-dimension search method, while the complexity is largely reduced.


\subsection {Performance of Algorithms}

\begin{figure}
\centering
\includegraphics[width=3.0in]{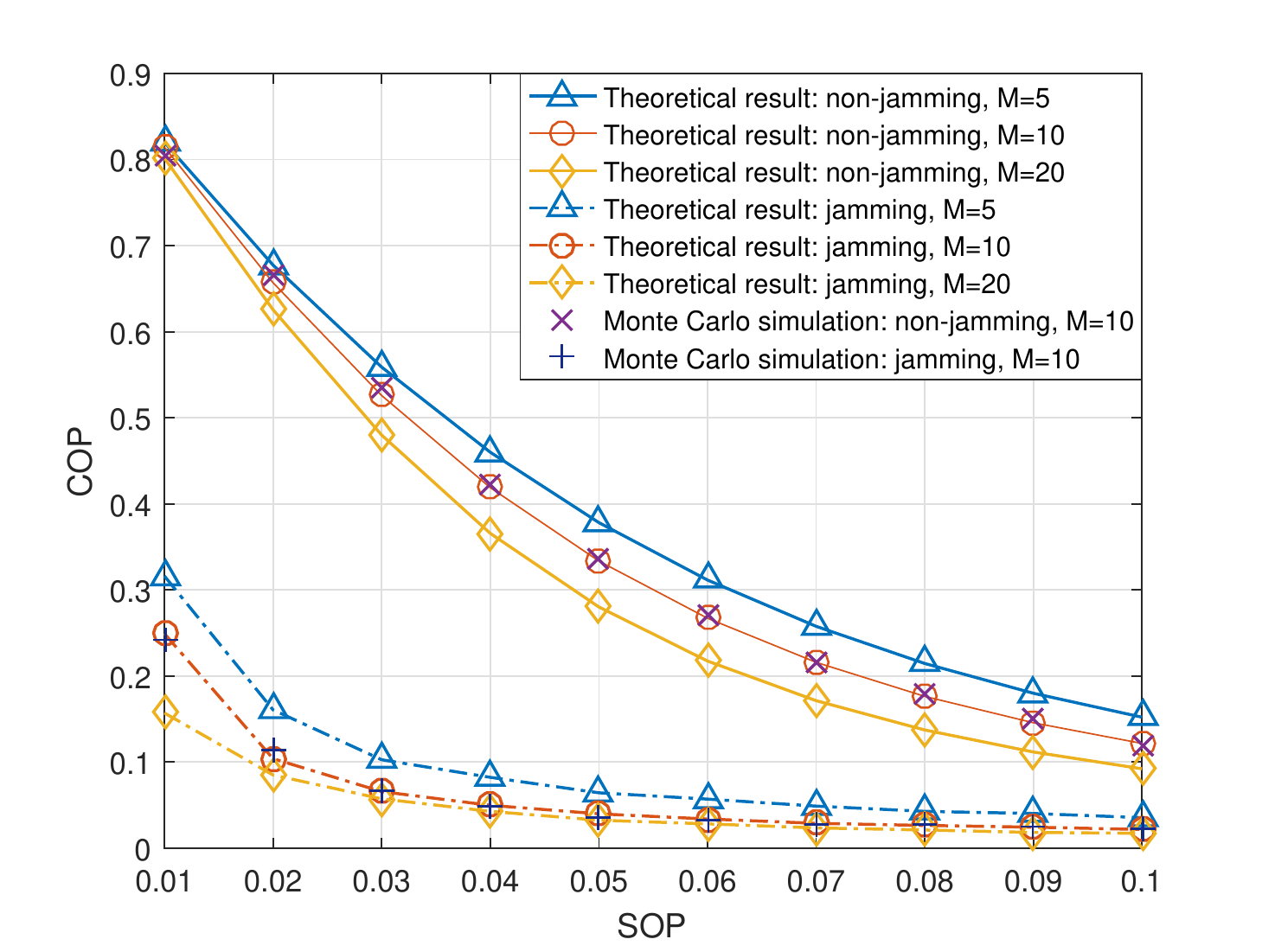}
\caption{COP versus SOP with the variation of the number of legitimate nodes.}
\label{nodes_51020}
\end{figure}

\begin{figure}
\centering
\includegraphics[width=3.0in]{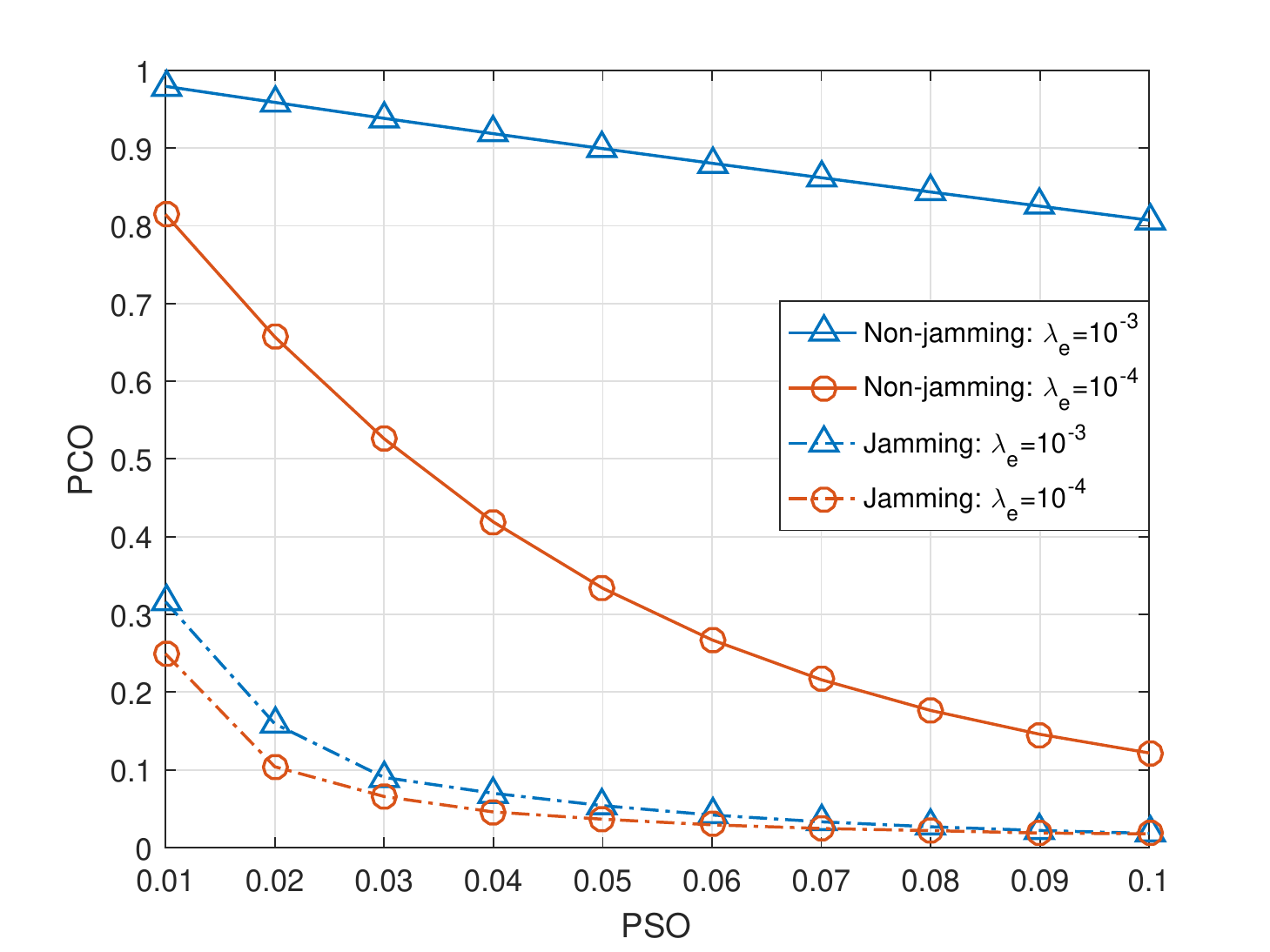}
\caption{COP versus SOP with the variation of the density of malicious eavesdroppers.}
\label{density_345}
\end{figure}

In this part, we discuss the secure performance under non-jamming and jamming conditions by comparing the results derived from Algorithm 1 and the polyblock approximation with one-dimension search method.
First, the correctness of the
secure routing and transmit power optimization
algorithms for non-jamming and jamming conditions
 are verified by comparing the theoretical results with the Monte Carlo simulation results.
The Monte Carlo simulation results are obtained through 10,000 simulation runs.
As Fig. \ref{nodes_51020} illustrates, the theoretical results match the Monte Carlo results for both circumstances with or without the assistance of friendly jamming, which verifies the correctness of the analytical results we derived.

Next, we study the performance of COP versus SOP for different values of $M$ in the network in Fig. \ref{nodes_51020}.
It shows that as the number of nodes increases, the connection outage probability decreases. Since the density of legitimate nodes in a certain area increases, the distance of each hop is shortened and the SNR at each receiving node increases. Thus it is less likely that the SNR is below the given threshold and the occurrence of connection outage is reduced.

The density of the eavesdroppers in the network $\lambda_e$ is varied and the results are plotted in Fig. \ref{density_345}.
Since the increase in the density of eavesdroppers  $\lambda_e$ exposes the legitimate nodes to more eavesdroppers, which increases the probability that the message may be interpreted by malicious eavesdropper(s), the SOP becomes higher as $\lambda_e$ increases for a fixed COP, as Fig. \ref{density_345} shows.
Besides, from Fig. \ref{nodes_51020} and Fig. \ref{density_345}, it is clear that
 the COP applying friendly jamming is much lower than that without the help of jammer under a fixed SOP requirement, which indicates the performance improvement with the assistance of jamming.


\begin{figure}
\centering
\includegraphics[width=3.0in]{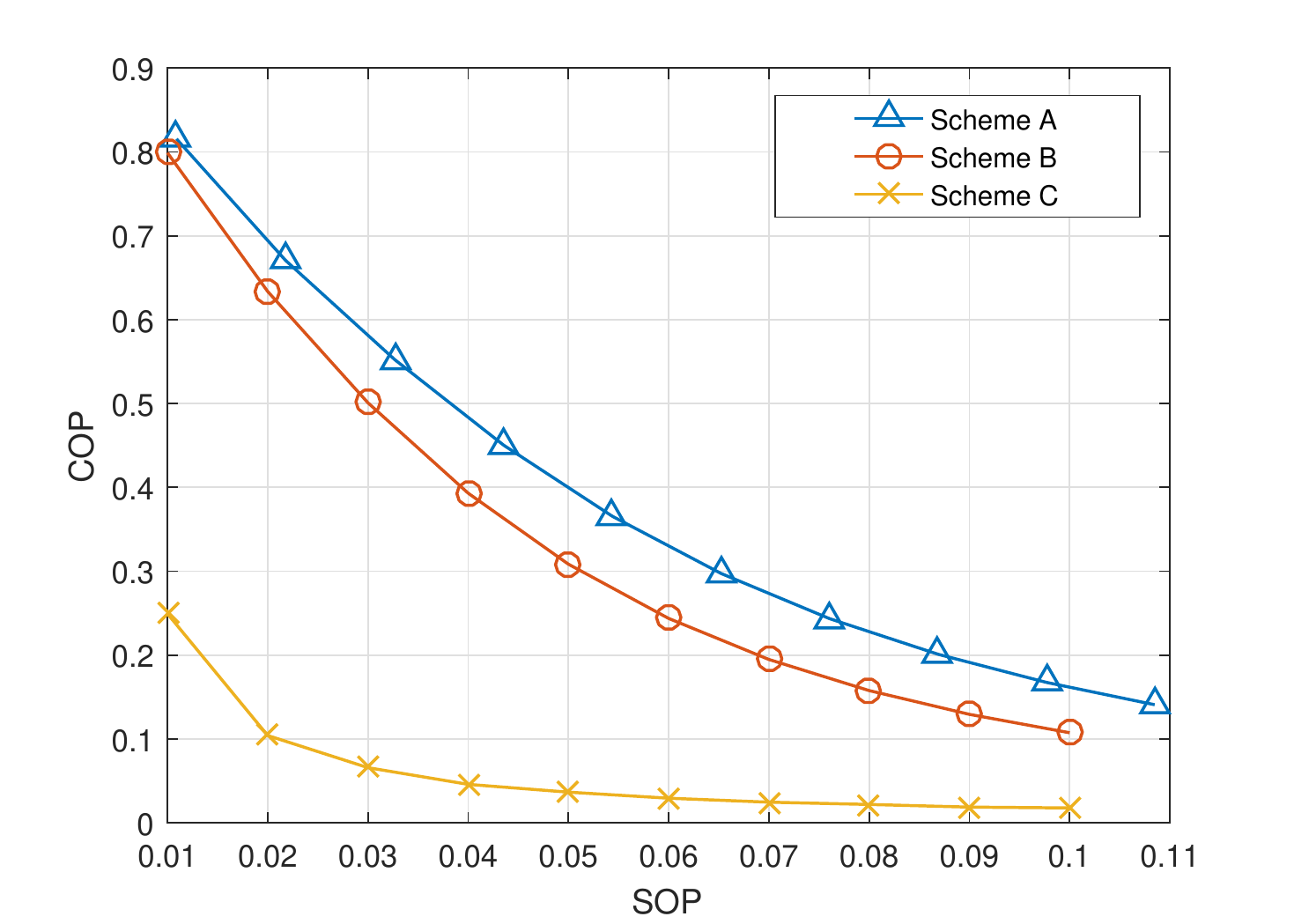}
\caption{COP versus SOP for three different schemes:
A: setting identical transmit powers;
B: Algorithm 1; 
C: jammer assisted tramsmission.}
\label{cmp_jam_and_nojammer}
\end{figure}


In order to further explain the performance improvement, we consider three different optimization schemes.
In scheme A, we use $\psi_n^{\frac{2}{2+\alpha}}$ as link weights for routing and set identical transmit powers for selected nodes. 
Scheme B also uses $\psi_n^{\frac{2}{2+\alpha}}$ as the weights and further optimizes the transmit powers, which represents our proposed Algorithm 1.
In scheme C friendly jamming is used to help improve the outage performance.
Given the value of the
maximum tolerable SOP $\zeta$, for scheme B and C, the transmit powers are optimized respectively;
keeping the sum of transmit powers in scheme A equals to that in scheme B, the transmit powers of each node using scheme A is derived and the corresponding achievable COP and SOP are obtained according to (\ref{pco_f}) and (\ref{pso_f}).
The SOP and COP outcomes for three schemes with the variation of $\zeta$ are plotted in Fig. \ref{cmp_jam_and_nojammer}.
The superior performance of scheme B compared to scheme A is clear, since the transmit powers for selected nodes in scheme B are optimized while that in scheme A are simply set equal.
In fact, the optimization of powers for transmit nodes balances the COP and SOP and avoids the cases that higher powers cause secure outage while lower ones lead to connection outage.
With the interference of artificial noise at the eavesdroppers, scheme C improves the connection performance remarkably, comparing to both schemes A and B.


\begin{figure}
\centering
\includegraphics[width=3.0in]{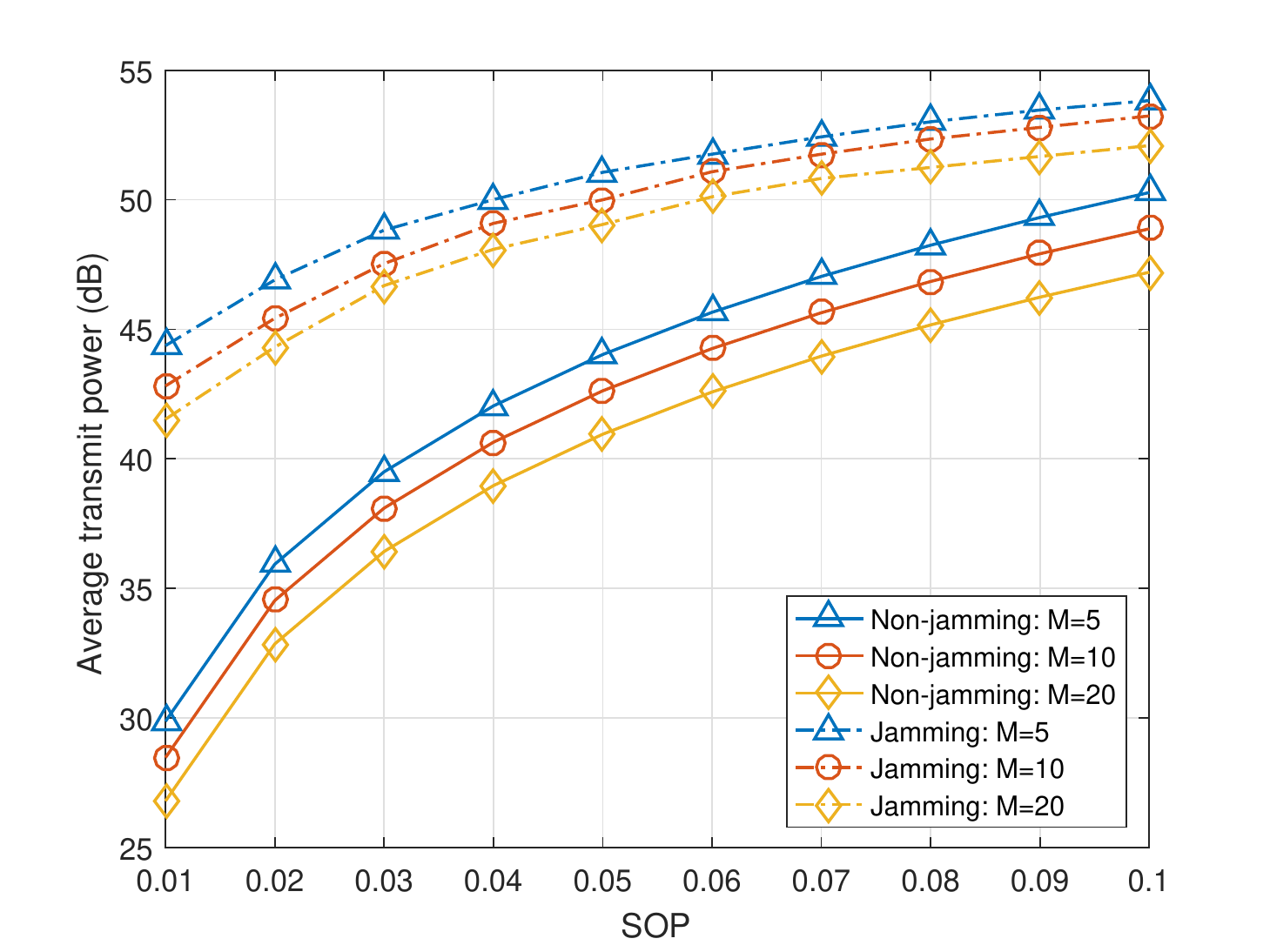}
\caption{Average transmit power versus SOP with the variation of the number of legitimate nodes.}
\label{power_node51020}
\end{figure}

\begin{figure}
\centering
\includegraphics[width=3.0in]{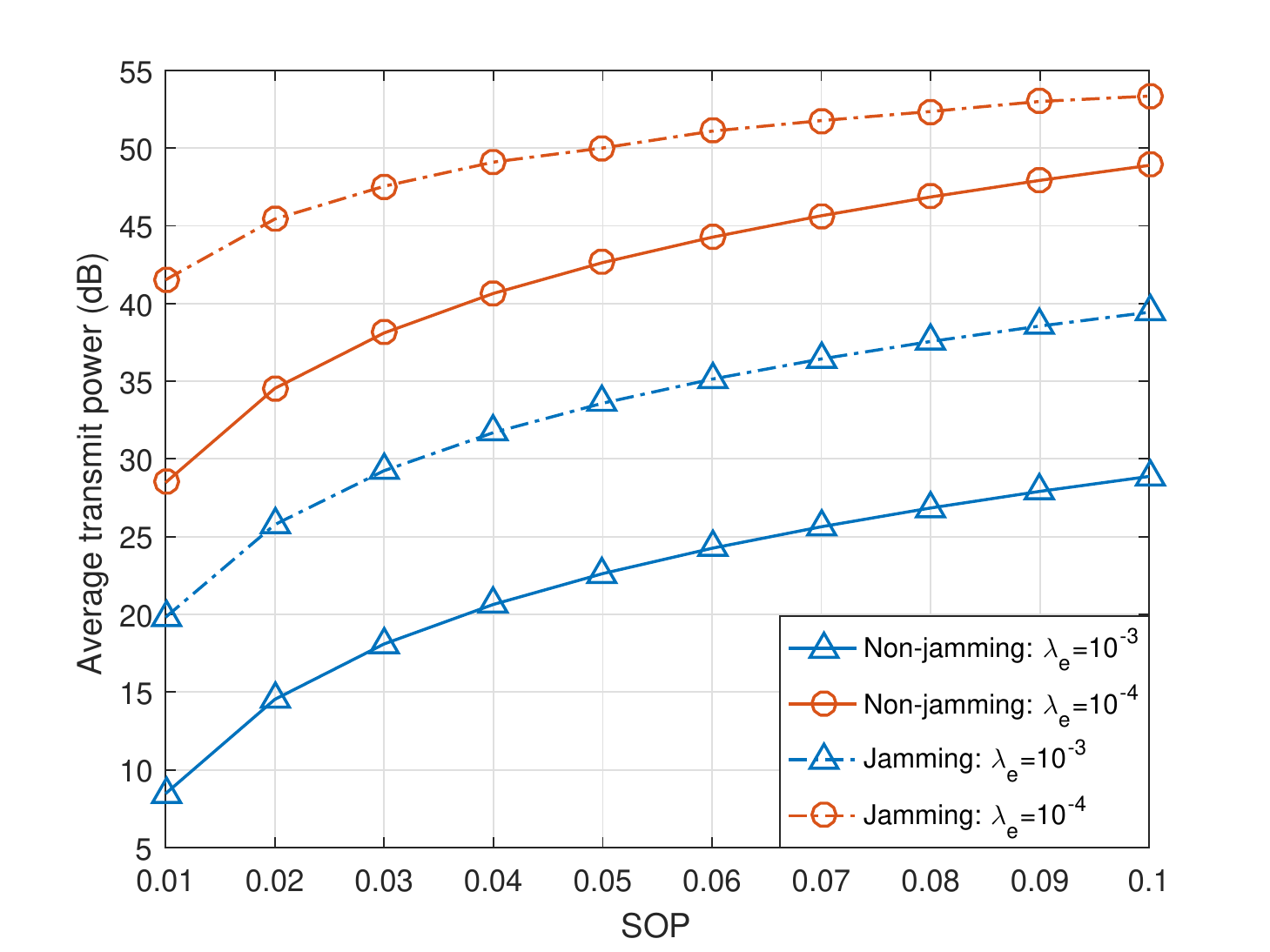}
\caption{Average transmit power versus SOP with the variation of the density of malicious eavesdroppers.}
\label{power_density543}
\end{figure}

\begin{figure*}
\centering
\subfigure[$M=10$]{
\begin{minipage}{0.3\textwidth}
\centering
\includegraphics[scale=0.8]{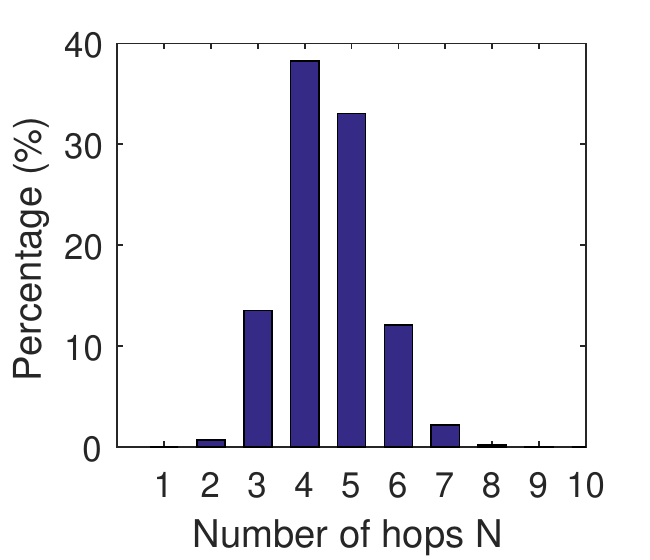}
\end{minipage}
}
\subfigure[$M=20$]{
\begin{minipage}{0.3\textwidth}
\centering
\includegraphics[scale=0.8]{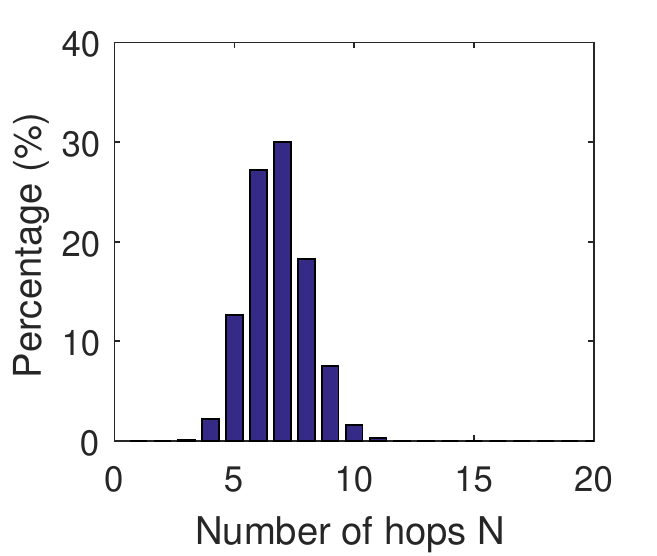}
\end{minipage}
}
\subfigure[$M=30$]{
\begin{minipage}{0.3\textwidth}
\centering
\includegraphics[scale=0.8]{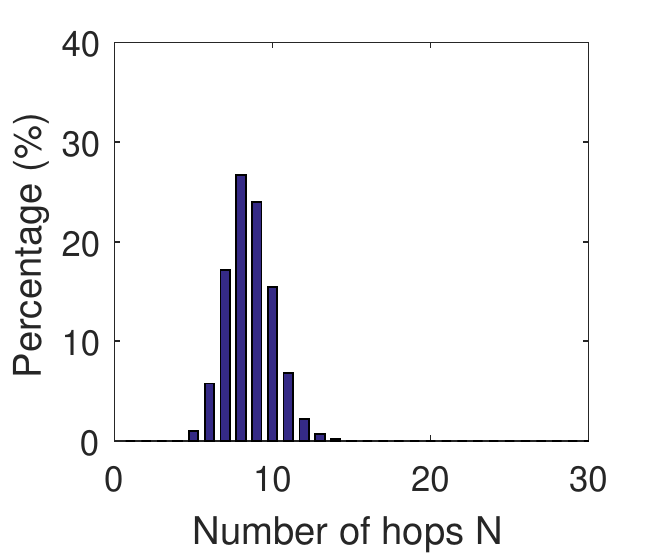}
\end{minipage}
}
\caption{The distribution of numbers of hops for different numbers of legitimate nodes $M$ in the network.}
\label{fig_tradeoff}
\end{figure*}
The variation of average transmit power of the legitimate transmit nodes (normalized by the noise power $\sigma^2$)
 with the SOP constraint is studied. The results under different values of $M$ and $\lambda_e$
  are depicted in Fig. \ref{power_node51020} and \ref{power_density543}, respectively.
We can observe that the average transmit power of the legitimate transmit nodes decreases with the increasing numbers of nodes in the network,
which is due to the shorter distance between the transmitter and receiver at each hop, so that the transmit power requirement for a success communication decreases.
Besides, the increase of the density of eavesdroppers leads to the decrease of transmit powers for legitimate nodes, so as to reduce the risks of being interpreted by malicious eavesdropper(s) for satisfying the SOP constraint.
Comparing the results under non-jamming and jamming conditions, it is obvious that the average transmit power of legitimate transmit nodes using friendly jamming is higher than that without the help of jammer. 
As a matter of fact, with the help of friendly jammer, security is guaranteed, thus the nodes can increase their powers to reduce communication outage.

\subsection{Trade-off Exploration}

We explore the trade-off between the numbers of hops and the transmit powers of legitimate nodes considering the routing security.
Given a pair of fixed source and destination nodes, we count the numbers of hops of the optimal secure routes obtained from 10,000 simulation runs. The distributions of the number of hops under different $M$ are plotted in Fig. \ref{fig_tradeoff}.
It can be summarized that the situations requiring a large
numbers of hops, e.g. $N>13$, under $M=30$, rarely happen, as they can lead to high probability of information leakage.
This indicates that more relaying nodes for help does not necessarily lead to a better security performance. Besides, the scenarios with very few numbers of hops rarely happen either, which usually require high powers to guarantee information transmission thus hard to satisfy the security requirement.
The results above reflect the trade-off between the numbers of hops and transmit power on routing security and indicate the importance of 
jointly optimizing transmit power and route selection for secure cooperative communication.

\section{conclusion}
In this paper, we considered the problem of secure routing for DF relaying ad-hoc networks with PPP distributed eavesdroppers.
We aimed to optimize the secure route and transmit powers which minimize the COP under a SOP constraint.
The closed-form expressions of COP and SOP for a given route were obtained and
the corresponding optimal transmit powers were optimized.
Besides, the routing weights were derived based on the expression of the transmit powers and the optimal secure route was obtained using the Dijkstra's algorithm.
We further studied the power optimization problem considering the application of friendly jamming. It was shown that this problem could be solved iteratively by the outer polyblock approximation with one-dimension search method.
On the other hand, considering the non-convexity of this problem, we 
proposed an iterative suboptimal algorithm based on the SCA algorithm.
Simulation results indicated that the SCA algorithm could provide a near-optimal solution with a lower complexity compared to the polyblock approximation with one-dimension search algorithm.
The performance improvement of the proposed algorithms for both non-jamming and jamming scenarios was demonstrated and the influence of network settings on the optimization outcomes was derived.
The discussion on the numbers of hops reflected the trade-off between the transmit powers and the routing selection.

\linespread{1.2}

\end{document}